\documentclass[journal,12pt,draftclsnofoot,onecolumn]{IEEEtran}
\usepackage[T1]{fontenc}
\usepackage[linesnumbered,ruled,vlined]{algorithm2e}
\usepackage{amssymb,amsmath}
\usepackage{caption}
\usepackage{subcaption}
\usepackage{comment}
\usepackage{color,soul}
\usepackage{multirow}
\usepackage{ragged2e} 
\usepackage{algorithmicx}
\usepackage{algpseudocode}
\usepackage{amsthm}
\usepackage{cite}
\usepackage{array}

\usepackage{cite}

\newcolumntype{L}[1]{>{\raggedright\let\newline\\\arraybackslash\hspace{0pt}}m{#1}}
\newcolumntype{C}[1]{>{\centering\let\newline\\\arraybackslash\hspace{0pt}}m{#1}}
\newcolumntype{R}[1]{>{\raggedleft\let\newline\\\arraybackslash\hspace{0pt}}m{#1}}
\newcommand{\be}{\begin{equation}}
\newcommand{\ee}{\end{equation}}
\newcommand{\bee}{\begin{eqnarray}}
\newcommand{\eee}{\end{eqnarray}}
\newcommand{\bse}{\begin{subequations}}
	\newcommand{\ese}{\end{subequations}}

\ifCLASSINFOpdf
   \usepackage[pdftex]{graphicx}
   \graphicspath{{../}{../}}
   \DeclareGraphicsExtensions{.pdf,.jpeg,.png,.jpg}
\else
  \usepackage[dvips]{graphicx}
   \graphicspath{{../}}
   \DeclareGraphicsExtensions{.eps}
\fi

\usepackage[caption=false,font=footnotesize]{subfig}
\newtheorem{propos}{Proposition}
\newtheorem{thm}{Theorem}
\newtheorem{defi1}{Definition}

\begin{document}
\title{Maximum-Quality Tree Construction for Deadline-Constrained Aggregation in WSNs }

\author{Bahram Alinia, Mohammad~H.~Hajiesmaili, Ahmad~Khonsari, and Noel Crespi
	\IEEEcompsocitemizethanks{\IEEEcompsocthanksitem A preliminary version of this paper \cite{Alinia15} has appeared in proceedings of \emph{IEEE International Conference on Computer Communications (INFOCOM)}, 2015. 
	
	B. Alinia and N. Crespi are with the department RS2M, Institut Mines-Telecom, Telecom SudParis, France. \mbox{E-mail: bahram.alinia@telecom-sudparis.eu, noel.crespi@mines-telecom.fr}.\protect
		\IEEEcompsocthanksitem M. H. Hajiesmaili is with Institute of Network Coding, The Chinese University of Hong Kong, Sha Tin, N.T., Hong Kong. \mbox{E-mail: mohammad@ie.cuhk.edu.hk}
		\IEEEcompsocthanksitem 	A. Khonsari is with School of Electrical and Computer Engineering, College of Engineering, University of Tehran, and School of Computer Science, Institute for Research	in Fundamental Sciences (IPM), Niavaran Sq., Tajrish Sq., Tehran, Iran, P. O. Box 19395-5746. \mbox{E-mail: ak@ipm.ir}
	}
	\thanks{}}


%

\maketitle

\begin{abstract}
In deadline-constrained wireless sensor networks (WSNs), quality of aggregation ($QoA$) is determined by the number of participating nodes in the data aggregation process. The previous studies have attempted to propose optimal scheduling algorithms to obtain the maximum $QoA$ assuming a fixed underlying aggregation tree. 
However, there exists no prior work to address the issue of constructing optimal aggregation tree in deadline-constraints WSNs. The structure of underlying aggregation tree is important since our analysis demonstrates that the ratio between the maximum achievable $QoA$s of different trees could be as large as $O(2^D)$, where $D$ is the deadline. This paper casts a combinatorial optimization problem to address optimal tree construction for deadline-constrained data aggregation in WSNs. While the problem is proved to be NP-hard, we employ the recently proposed Markov approximation framework and devise two distributed algorithms with different computation overheads to find close-to-optimal solutions with bounded approximation gap. To further improve the convergence of the proposed Markov-based algorithms, we devise another initial tree construction algorithm with low computational complexity. Our extensive experiments for a set randomly-generated scenarios demonstrate that the proposed algorithms outperforms the existing alternative methods by obtaining better quality of aggregations. 
\end{abstract}

\IEEEpeerreviewmaketitle

\section{Introduction}
\label{sec_intro}
\subsection{Motivation}
\IEEEPARstart{N}{owadays}, monitoring and tracking applications are intrinsically intertwined with a plethora of Wireless Sensor Networks (WSNs). To accomplish monitoring effectively, data gathering is witnessed as a fundamental operation in such applications. 
In addition, limited battery of sensors emphasizes the need for energy-aware data gathering design. 
However, packet transmission as the major source of energy depletion turns energy conservation in data gathering into an acute problem~\cite{Heinzelman}. 

To reduce the energy depletion of sensors due to excessive packet transmission, \emph{data aggregation} \cite{Heinzelman,Rajagopalan} has been proposed as a promising energy conservation mechanism to eliminate the necessity of redundant transmission. In a typical data aggregation scenario, a \textit{data aggregation tree} is constructed over the underlying WSN topology \cite{Shroff} and some intermediate nodes are solicited to aggregate/fuse the gathered data of different sensors by in-network computation and transmit a single packet to the next hop. Consequently, the amount of packet transmission is significantly reduced, thereby the overall energy consumption decreases. 

Despite its apparent benefits in reducing overall energy usage, data aggregation, however, can impose additional delay toward sending the data to the sink since the intermediate nodes in aggregation tree must \textit{wait} to gather enough data from the predecessors and then aggregate and forward the data to the next hop. This additional data aggregation waiting time might be intolerable for many real-time and surveillance applications that are sensitive to the amount of the latency of the receiving data. For example in target tracking application, the detected location of a moving object may has perceptible error with the actual location if data aggregation process takes too long \cite{Shroff2013}. Thus, the delay contribution of a data aggregation algorithm must be taken into account in an efficient design so as to respect the delay constraint of the application.

Some previous researches have considered participation of all sensor nodes in data aggregation and aimed to minimize the aggregation delay as the objective \cite{Li,Xu}. 
However, participation of all sensor nodes introduces severe interference and may lead to terminating data aggregation in a time that is beyond the application's tolerable delay even attempting in delay minimization. Consequently, these designs fail to guarantee a maximum application-specific tolerable delay for data aggregation. 

\subsection{Deadline-Constrained Data Aggregation Scenario and Challenges}
As a promising alternative, the idea of deadline-constrained data aggregation has been advocated in the recent studies \cite{Shroff,Alinia}. The general idea is to incorporate a maximum application-specific tolerable delay, namely \textit{deadline}, as a hard constraint, and try to improve the \emph{Quality of Aggregation} ($QoA$) by participating the sensor nodes as much as possible before the aggregation time exceeds the deadline. Subsequently, the problem turns into maximizing $QoA$, subject to the application-specific deadline constraint~\cite{Shroff,Alinia}. Toward this goal, the following two critical challenges should be addressed appropriately: $1)$ the scheduling policy, and $2)$ the structure of aggregation tree.

\textbf{(1) Scheduling policy.} Delay in data aggregation is originated from two sources: (i) waiting time to gather the data of predecessor nodes in data aggregation tree, and (ii) waiting time due to the interference issue which is an inherent challenge in wireless networks. 
In this way, if overall waiting time of a node exceeds a specific value, its data cannot be delivered to the sink before the deadline. Devising an efficient policy that schedules the nodes' waiting time while preventing degradation of the $QoA$ and meeting the delay constraint of the application is a challenging problem.  
The previous research has tried to find an efficient scheduling such that $QoA$ is maximized \cite{Shroff}. The details are explained in Section.~\ref{sec:rel}. 

\textbf{(2) The structure of aggregation tree.} The structure of data aggregation tree is another important factor such that the number of participant nodes in data aggregation could be further improved by constructing a proper data aggregation tree. 
  For data aggregation, a tree rooted at the sink node is the common structure since it simplifies design of routing and aggregation protocols and also helps to avoid problems such as double counting \cite{Nath}. Without constructing an appropriate aggregation tree, we may not be able to achieve a desired level of $QoA$ even by designing the scheduling algorithm optimally. We show that even with optimal scheduling policy, the structure of underlying aggregation tree can make a big difference on achievable $QoA$.

While the first critical challenge on how to schedule the sensors is tackled in the previous studies \cite{Shroff,Alinia}, the ultimate optimal design, however, cannot be fully achieved without taking into account the second critical challenge on how to construct an optimal underlying data aggregation tree. The goal of this paper is to address the second challenge. 

\subsection{Summary of Contributions}
In this paper, we aim to construct an optimal aggregation tree and run an optimal tree-specific
scheduling algorithm on the tree to maximize $QoA$. However, constructing the optimal aggregation tree given the network topology is nontrivial even in centralized manner. This is more problematic when we seek an appropriate solution amenable to distributed realization so as the sensor nodes choose their parents just using local information. We address this problem in single-sink WSNs setting through the following contributions:

\begin{itemize}
\item We demonstrate the impact of data aggregation tree structure on $QoA$ by theoretical discussion and explanatory example. We show that the ratio between the maximum achievable $QoA$s of two data aggregation trees is $O(2^D)$ in the worst case where $D$ is the aggregation deadline. This observation makes the problem of constructing maximum $QoA$ tree intriguing. Besides, we prove that the problem of optimal tree construction belongs to the class of NP-hard problems.

\item After formulating the underlying tree construction problem, we leverage the recently proposed Markov approximation framework \cite{Chen2} as a general framework toward solving combinatorial network problems in distributed fashion. By addressing the unique challenges devoted to our problem, we devise two close-to-optimal algorithms in which the sensor nodes contribute to migrating toward the optimal tree in an iterative manner. The highlights are distributed implementation, bounded approximation gap, and robustness against the error of global estimation of sensor nodes by local information.

\item  To further improve the obtained $QoA$ and convergence of Markov approximtion framework, we propose an initial tree construction algorithm, called \textit{FastInitTree}, as the initialization step of our main algorithm.
The algorithm features are low computational complexity and close-to-optimal estimate for the initial tree construction, i.e., the $QoA$ obtained by initial tree constructed by \textit{FastInitTree} is close to the $QoA$ achieved when the main algorithm converges.  We also analyze the validity of constructed trees when the deadline value is changed by the application and we prove that for the cases that the deadline decreases there is no need to executed the algorithms again and construct a new tree. 


\item Through extensive experiments, we show the superiority of our algorithms and compare them to the previous work\cite{Shroff}. Results demonstrate that four different versions of our algorithm (two algorithms which start from random tree and two algorithms which use  \textit{FastInitTree}'s output as initial Markov chain state) greatly increase the $QoA$ of the proposed algorithm in \cite{Shroff} (that is only by finding optimal scheduling without optimal tree construction).
\end{itemize}

\subsection{Paper Organization}
The rest of this paper is organized as follows. We review the related work in Section~\ref{sec:rel}. In Section~\ref{sec_SystemModel}, the system model is introduced and by motivating examples and theoretical analysis the impact of aggregation tree on $QoA$ is investigated. Problem formulation and NP-hardness analysis are explained in Section~\ref{sec_OptimizationProblem}. In Section~\ref{sec:markov}, we devise two distributed algorithms for the problem. In Section~\ref{sec:init}, we explain our tree initialization algorithm. Simulation results are described in Section~\ref{sec:sim}. Finally, concluding remarks and future directions are mentioned in Section~\ref{sec:conc}. 

\section{Related Work}
\label{sec:rel}
\subsection{Minimum delay and deadline-constrained aggregation}
The problem of minimum delay data aggregation tackled intensively in the literature.
In \cite{Chen}, it is proved  that the Minimum Latency Aggregation Scheduling (MLAS) problem
is NP-hard and a $(\Delta -1)$-approximation algorithm has been presented where $\Delta$ is the maximum node degree 
in the network. The current best approximation algorithms in \cite{Wan,Xu} achieve an upper bound of $O(\Delta +R)$ on data aggregation delay where $R$ is the network radius. 
While most studies consider a protocol interference model, the studies in \cite{Li2,Li} assume a physical interference model that is more practical than the former. 
In \cite{Li2}, a scheduling algorithm for tree-based data aggregation is designed that achieves a constant approximation ratio by bounding the delay at $O(\Delta +R)$. The work is extended in \cite{Li} for any arbitrary network topology. A connected dominating set or maximum independent sets are employed in \cite{Guo} to provide a latency bound of $4R^\prime +2\Delta -2$ where $R^\prime$ is inferior network radius.
 
Within the context of deadline-constrained data aggregation models, the goal is not to minimize the delay as an objective of the problem. Rather, the objective is to maximize the number of  sensor nodes participating in aggregation while respecting the application-specific deadline.
This type of real-time data aggregation has recently gained attention in some works  \cite{Shroff2014,Shroff2013,Shroff,Shroff2011,Alinia}. In this regard, \cite{Shroff} presented a polynomial time optimal algorithm for the problem under the deadline and one-hop interference constraints.
The problem is extended in \cite{Shroff2013} for a network with unreliable links under an additional constraint on nodes' energy level. In \cite{Shroff2013}, the authors proved that in a network with $V$ nodes, the problem is NP-hard when the maximum node degree of the aggregation tree is $\Delta$. They proposed a polynomial-time exact algorithm when $\Delta=O(\log V)$. In \cite{Alinia}, the authors considered the same problem of \cite{Shroff} by taking into account the effect of data redundancy and spatial dispersion of the participants in the quality of final aggregation result and proposed an approximate solution for proved NP-hard problem. In a more general case, \cite{Shroff2014} tackles the utility maximization problem in deadline constrained data aggregation and collection then provides efficient approximation solutions. A main drawback of the aforementioned studies is that they all have tried to maximize the quality of data aggregation on a given tree and neglect the impact of the data aggregation tree structure.

\subsection{Optimum Aggregation Tree Construction}
Several studies have tackled the problem of constructing optimal data aggregation tree \cite{Tan,Wu,Li3,Fahmy,Kuo} where all have been shown to be NP-hard. 
In \cite{Fahmy}, the problem of maximum lifetime aggregation tree is studied for single sink WSNs. 
\cite{Wu} extends the problem for multi-sink WSNs.
The problem of constructing an aggregation tree in order to minimize the total energy cost is addressed in \cite{Kuo}. As solution, a constant factor approximation algorithm is proposed. 
In  \cite{Tan}, the problem of constructing a minimum cost aggregation tree under 
Information Quality (IQ) constraint has been tackled. The authors considered event-detection WSNs and defined IQ as detection accuracy. \cite{Shan} shows that for the shortest path trees, the problem of building maximum lifetime data aggregation tree can be solved in polynomial time and propose two centralized and distributed algorithms.
In this paper, however, we aim to construct maximum quality aggregation tree under deadline constraint. This problem has not been addressed yet by the research community. Moreover, our solution method in solving the problem is completely different from the previous research and is based on a recently-proposed theoretical foundation, namely Markov approximation that may be considered as a potential solution for the same category of problems.

\section{System Model and Problem Motivation}
\label{sec_SystemModel}

\begin{table}[!t]
	\begin{center}
		\begin{tabular}{|c|L{12cm}|}
			\hline
			\textbf{Notation} & \textbf{Definition} \\
			\hline \hline
			$\mathcal{V}$ & Set of sensor nodes, $V \triangleq |\mathcal{V}|$ \\\hline
			$\mathcal{T}(\mathcal{G})$ & Set of all spanning trees in graph $\mathcal{G}$ \\\hline
			$D$ &  Sink deadline \\\hline
			$H^\psi(i)\subseteq \mathcal{V}$ & The set that consists of node $i$ and all its predecessors (except the sink) in tree $\psi$\\\hline
			$F_i$ & $F_i=1$, if node $i$ is a source, $F_i=0$, otherwise\\\hline
			$n_i^\psi$ & $n_i^\psi=1$, if node $i$ in tree $\psi$ is allowed to send data to its parent, $n_i^\psi=0$, otherwise (${\vec n^\psi=[n^\psi_i,i\in \mathcal{V}]}$) \\\hline
			$W_i^\psi$ & Waiting time of \emph{participant} node $i$ in aggregation tree $\psi$, ($\vec W^\psi=[W^\psi_i,i\in \mathcal{V}]$) \\\hline
			 $QoA_{\psi}(\vec{W})$ & The $QoA$ in tree $\psi$ and deadline $D$ and assigned waiting times determined by $\vec{W}$\\\hline		
		\end{tabular} 
	\end{center}
	\caption{Summary of key notations}
	\label{tbl:not}
\end{table}
\subsection{WSN System Model}
\label{sec:sys}
Consider a WSN whose topology is a graph ${\mathcal{G}=(\mathcal{V} \cup \{S\},\xi)}$ where $S$ is the sink node, $\mathcal{V}$ is the set of sensor nodes with $|\mathcal{V}| \triangleq V$, and $\xi$ is the set of links between sensor nodes. We assume that all nodes have a fixed communication range and $(i,j)\in\xi$ if nodes $i$ and $j$ are adjacent, i.e., they are in the communication range of each other. Without loss of generality, we assume that each link has a unit capacity. Moreover, we suppose that the system is time-slotted and synchronized and a transmission takes exactly one time slot. In deadline-constrained scenario, the data has to be received by the sink by the end of at most $D$ time slots, where the value of $D$ is specified by the deadline requirement of the applications. The interference model is one-hop such that two active links having a node in common cause an interference \cite{Shroff,Shroff2011,Alinia}.

The data aggregation is done using an overlay spanning tree $\psi\in\mathcal{T}(\mathcal{G})$ (rooted at the sink node) on top of the underlying WSN topology where $\mathcal{T}(\mathcal{G})$ is the set of all spanning trees in the graph $\mathcal{G}$. Let $H^\psi(i)\subseteq \mathcal{V}$ be the set that consists of node $i$ and all its predecessors (except the sink) in aggregation tree $\psi$.  

We consider two types of nodes, source nodes and relay nodes. Source nodes can sense their own data and forward/aggregate the other nodes' data. Relay nodes just forward/aggregate the data of other nodes. To illustrate this, we use binary variable $F_i$, where $F_i=1$ if node $i$ is a source and $F_i=0$, otherwise.
Moreover, binary variable $n_i^\psi$ with $n_i^\psi=1$ indicates that node $i$ in tree $\psi$ is allowed to send data to its parent and ${\vec n^\psi=[n^\psi_i,i\in \mathcal{V}]}$. Indeed, $n_i^\psi=1$ indicates that node $i$ participates in data aggregation. In this case, if $F_i=1$ then node $i$ is a source participant, otherwise node $i$ participates in data aggregation as a relay node, i.e., it just aggregates the received data from its successors and forwards to its parent.

Let ${\mathcal{V}_{\textrm{leaf}}^\psi\subseteq \mathcal{V}}$ be the set of all leaf nodes and 
$\mathcal{V}_{\textrm{sel-src}}^\psi\subseteq \mathcal{V}$ be the set of source nodes \emph{selected} for data aggregation in tree $\psi$. Indeed, $i\in \mathcal{V}^\psi_{\textrm{sel-src}}$, if $i$ is a source and all of its predecessors are selected for aggregation, i.e., $\mathcal{V}_{\textrm{sel-src}}^\psi=\left\{i\in \mathcal{V}: F_i=1\textrm{ and }\prod_{j\in H^\psi(i)}n_j^\psi=1\right\}.$

To devise a feasible aggregation scheme, we assign a waiting time of $W_i^\psi, 0\leq W_i^\psi\leq D$ time slots to each \emph{participant} node $i$ in aggregation tree $\psi$ and $\vec W^\psi=[W^\psi_i,i\in \mathcal{V}]$. The notion $QoA_{\psi}(\vec{W})$ denotes the quality of aggregation in tree $\psi$ under assigned waiting times determined by $\vec{W}$. We follow the definition in \cite{Shroff} for $QoA_{\psi}(\vec{W})$ as the number of source nodes that participate in data aggregation, i.e.,
\begin{equation}
\label{eq:QoA}
QoA_{\psi}(\vec{W})={|\mathcal{V}^\psi_{\textrm{sel-src}}|}={\sum_{i\in \mathcal{V}} F_i\prod_{j\in H^\psi(i)} n^\psi_j}.
\end{equation}

For notational convenience, we define $QoA_{\psi}(\vec{W}_i)$ as $QoA$ of the sub-tree of tree $\psi$ rooted at node $i$ with assigned waiting time of $W_i$.
Hereafter, we use $QoA$ and $W_i$ instead of $QoA_{\psi}(\vec{W}_i)$ and $W_i^\psi$ when the corresponding tree and scheduling are obvious, or a specific tree or scheduling is not the matter of concern. The summary of notations are listed in Table~\ref{tbl:not}.

\subsection{Optimal Scheduling Policy}
Now, we proceed to explain by example how to find the maximum $QoA$ in a given tree in simple and tractable cases. 
The examples also clarify the data aggregation model for the reader.

\begin{figure}[t!]
	\centering
	\includegraphics[width=.4\textwidth]{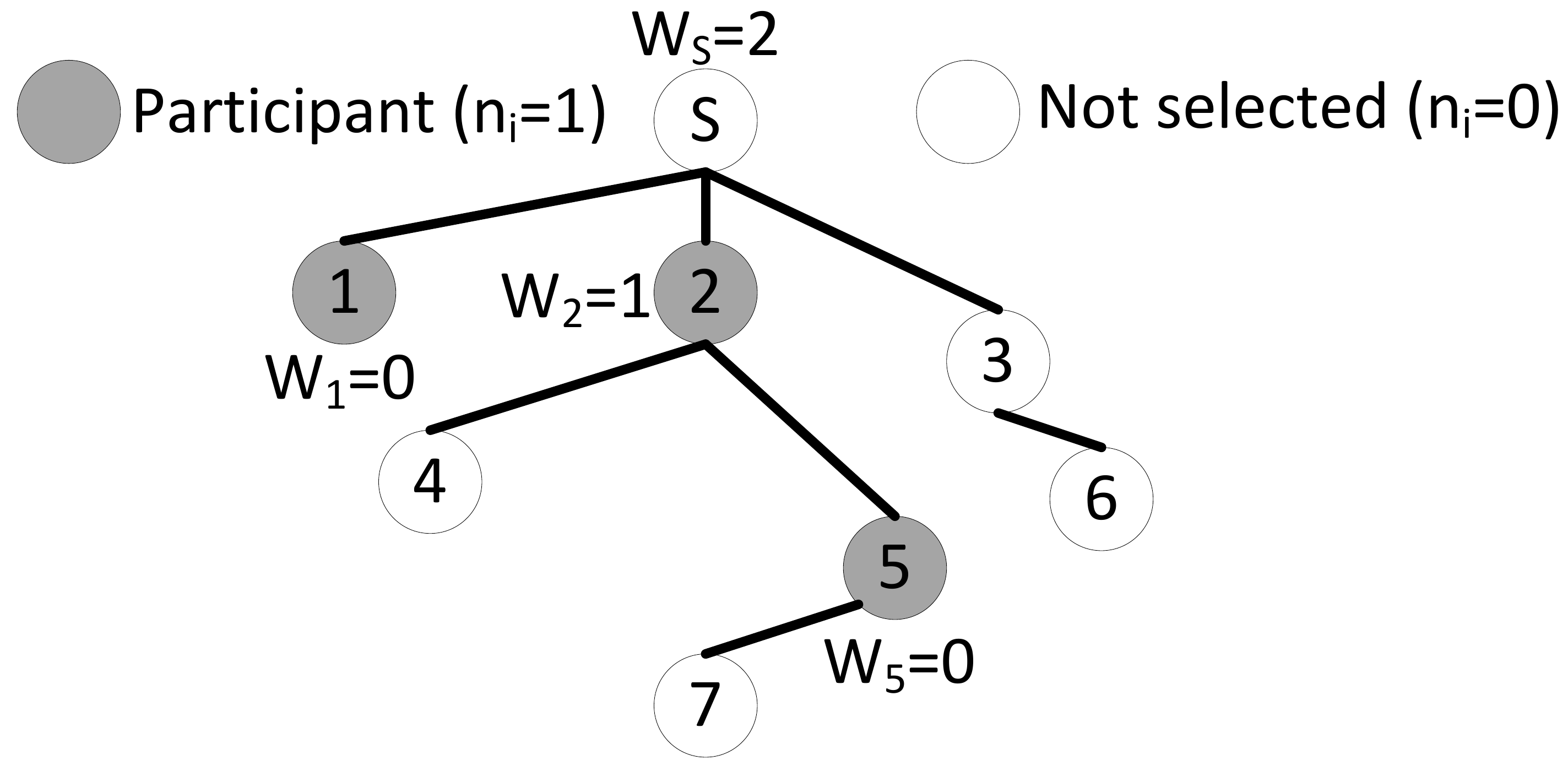}
	\caption{Example 1. Optimal scheduling in a fixed underlying tree with $D=2$.}
	\label{fig:Ex1}
\end{figure}

\textbf{Example 1. (Calculating the maximum $QoA$ given a fixed aggregation tree):} Consider the data aggregation tree in Fig.~\ref{fig:Ex1} where 
the sink deadline is set to $D=2$ and all nodes are source.
With the given deadline, the sink can choose at most $D=2$ children (due to one-hop interference constraint \cite{Shroff}) and assign their waiting times as \emph{distinguished} values between $0$ and $D-1=1$. To maximize the number of source participant nodes ($QoA$), one of the possible choices for the sink is the assignment of $W_1=0$ and $W_2=1$. With this assignment, node $2$ can assign a waiting time of $0$ to one of its children (in this example node $5$ with $W_5=0$). Eventually, the maximum $QoA$ is $3$ and participant nodes are $1, 2, 5$. During the aggregation process, in the first time slot, node $1$ and node $5$ send their packets to their parents in parallel. In the second time slot, node $2$ aggregates its own packet with the received data from node $5$ and sends the aggregated data to the sink. It is not hard to see that this scheduling policy is optimal, i.e., it achieve the maximum $QoA$ given the fixed aggregation tree.  
As a non-optimal waiting time assignment, consider the assignment of $W_1=1, W_2=0$. In this case, the final $QoA$ is $2$ with participant nodes $1$ and $2$. With $D=3$, the maximum $QoA$ is $7$ and the optimal assignment is
$W_1=W_4=W_6=W_7=0, W_2=2, W_3=1$ and $W_5=1$.
In \cite{Shroff}, an algorithm is proposed to achieve the maximum $QoA_{\psi}(\vec{W})$ in a given tree $\psi$. 
The scheduling algorithm in \cite{Shroff} is optimal given a fixed tree as input and it does not change the structure of the tree for further improvement of $QoA$.

\subsection{Investigating The impact of Aggregation Tree}
We argue that the aggregation tree structure may significantly impact $QoA$. Firstly, we claim that the optimal aggregation does not follow any regular pattern. For example, the structure of the optimal tree should not follow \textit{chain-like} long trees. The reason is that when sink imposes a deadline $D$, all nodes with height ``$\geq D$'' cannot participate in data aggregation due to the delay constraint. Consequently, the height of the tree is limited to $D$ and long trees cannot be proper structures. Instead, one might suggest a tree so as the height of the majority of nodes is ``$\leq D$''. But, the waiting time of a node with height $h$ is upper bounded by $D-h$ and hence it can choose at most $D-h$ children of itself as the participants. The others together with their successors are ignored. Thus, same as the long tree, a \textit{star-like} fat tree may yield a non-optimal $QoA$. \emph{Roughly speaking, an aggregation tree which is neither so long nor so fat is suitable.}
It is important to stress that the above conditions cannot bring significant insights to devise an algorithm to construct the optimal tree. In the next example, we demonstrate that maximum $QoA$ of two aggregation trees of a same network can be large even using optimal scheduling of \cite{Shroff}.

\begin{figure*}[t!]
	\centering
	\vspace{-0.05em}
	\begin{subfigure}[b]{0.3\textwidth}
		\begin{center}
			\includegraphics[width=\textwidth]{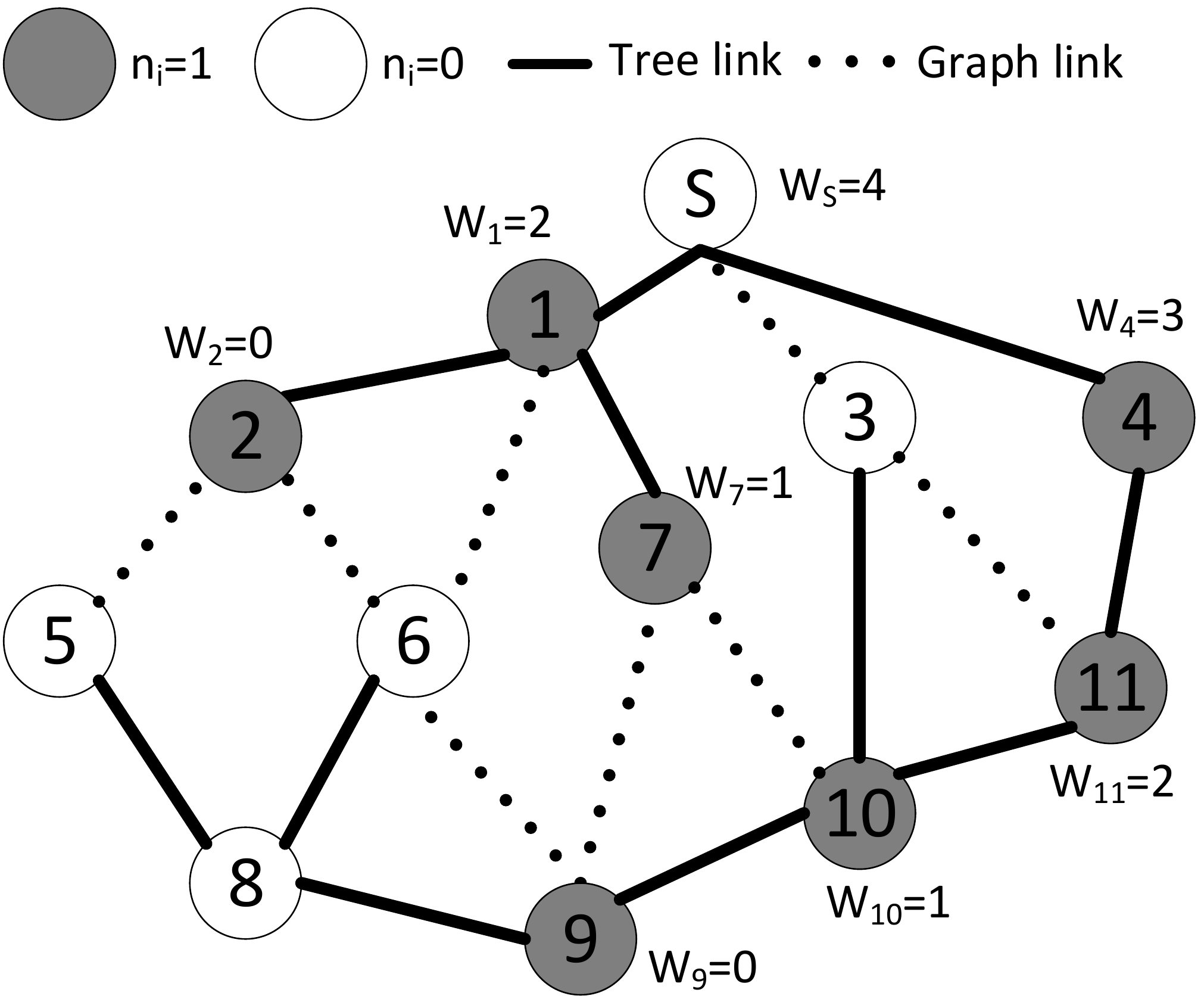}
			\caption{Long tree: $QoA=7$}
			\label{fig:Ex2a}
		\end{center}
	\end{subfigure}%
	~ 
	\begin{subfigure}[b]{0.3\textwidth}
		\begin{center}
			\includegraphics[width=\textwidth]{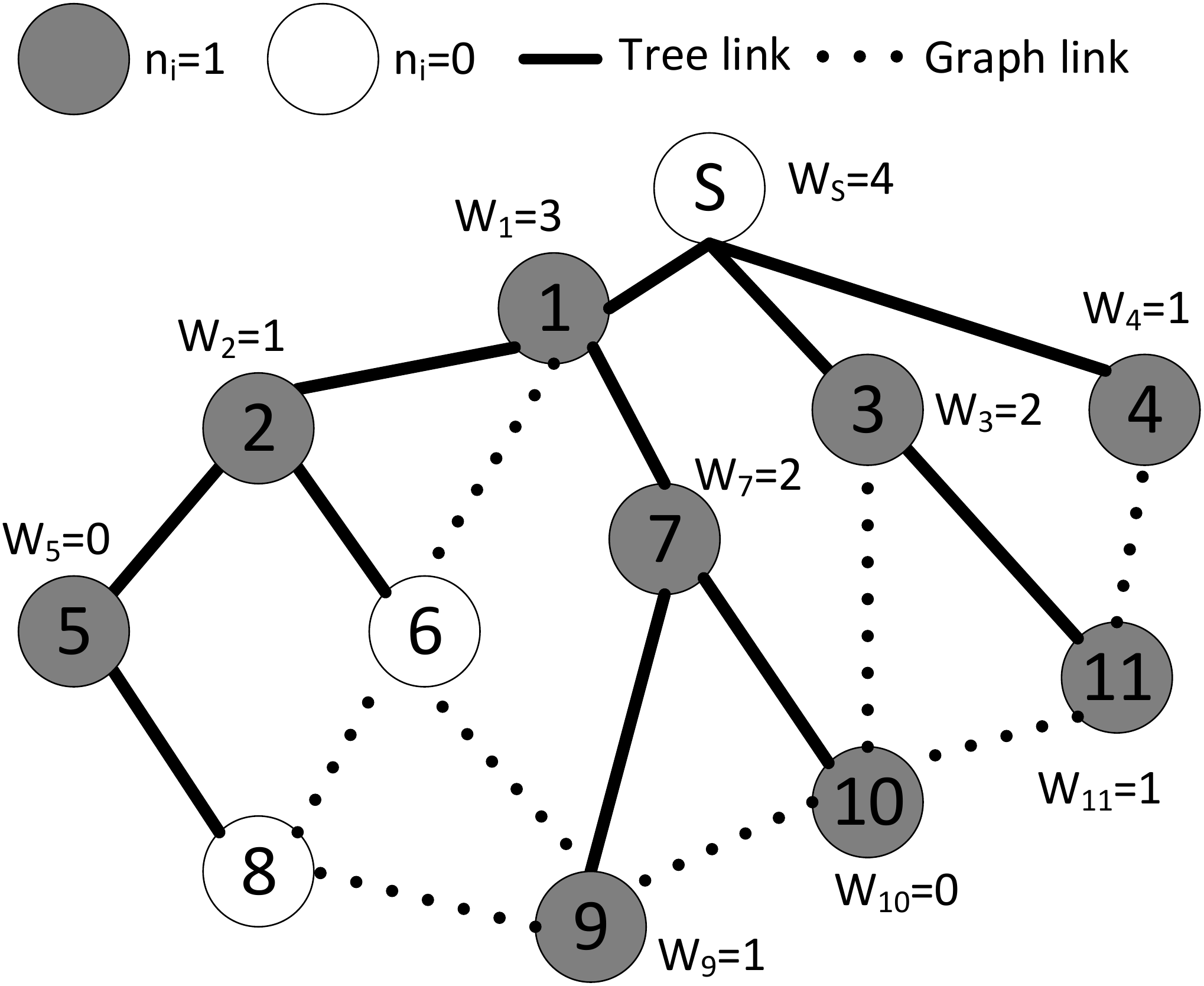}
			\caption{Random tree: $QoA=9$}
			\label{fig:Ex2b}
		\end{center}
	\end{subfigure}%
	\begin{subfigure}[b]{0.3\textwidth}
		\begin{center}
			\includegraphics[width=\textwidth]{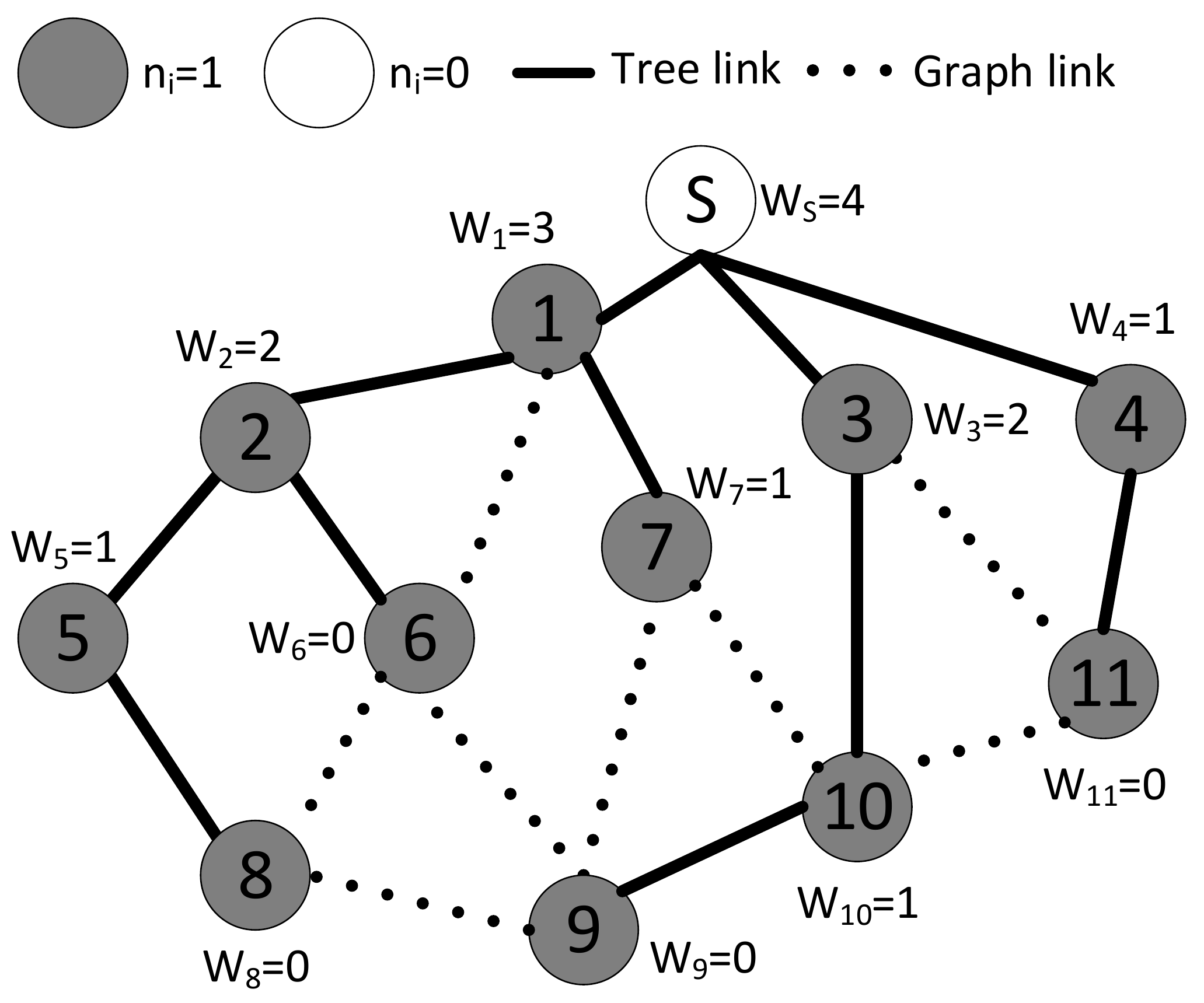}
			\caption{Optimal tree: $QoA=11$}
			\label{fig:Ex2c}
		\end{center}
	\end{subfigure}%
	\caption{Impact of aggregation tree structure on the maximum $QoA$. The waiting times are computed using the optimal scheduling algorithm \cite{Shroff}. } 
	\label{fig:Ex2}
\end{figure*}

\textbf{Example 2. (Investigating the impact of aggregation tree on maximum achievable $QoA$):} Fig. \ref{fig:Ex2} illustrates the maximum achievable $QoA$ of three different data aggregation trees given a fixed underlying WSN topology. Fig.~\ref{fig:Ex2a} is an example of long tree. With sink deadline $4$, at most one node in height $4$ of aggregation tree can participate in data aggregation. That is, just one of the nodes $3$ and $9$, both with height $4$, can participate in the aggregation.  Moreover, the set of nodes $\{5, 6, 8\}$ are in a distance greater than $D$ and there is no way to participate them. In other words, this particular long tree structure already has no way to participate at least $4$ nodes in aggregation process. It turns out that the maximum $QoA$ of tree in Fig. \ref{fig:Ex2a} using the optimal scheduling \cite{Shroff} is $7$. 
Fig. \ref{fig:Ex2b} shows a random tree with the maximum $QoA$ of $9$. Finally, the optimal data aggregation tree is shown in Fig. \ref{fig:Ex2c} where all nodes are participants. The optimal tree in this toy example is obtained by trial and error. We emphasize that finding the optimal aggregation tree is not straightforward even in our tractable topology with only $12$ nodes, while in practice the scale of the network is much larger than that of this example.

\begin{thm}
\label{lemma_tree_importance}
For an imposed deadline $D$ where all nodes are source, the maximum values of $QoA$ in the optimal tree and worst-case tree are $2^D-1$ and $D$, respectively.
\end{thm}
\begin{proof}
It is proved in \cite{Alinia} that $QoA$ is bounded to $2^D-1$ regardless of the aggregation tree structure. The bound is touchable when the network graph is dense enough where an obvious case is a complete graph (for more details, refer to Section \ref{sec:init}). Therefore, we proceed to calculate the upper bound in the worst case. Indeed, the worst case occurs when we construct a chain-like tree with sink as the head of the chain. Observe that for a node $i$, $|H^\psi(i)|$ is equal to the distance of $i$ to the sink in aggregation tree $\psi$. 
In a chain tree, there is only one possible way of scheduling where each node $i$ having the property $|H^\psi(i)|\leq D$ assigned a waiting time of $D-|H^\psi(i)|$ and is a participant. There are $D$ such nodes and therefore the maximum $QoA$ of the tree is $D$. 
\end{proof}

The motivating example 2 and Theorem~\ref{lemma_tree_importance} clearly signify the importance of aggregation tree structure on the final $QoA$. In the next section, we formally formulate the optimal aggregation tree construction as an optimization problem. 

\section{Problem Formulation}
\label{sec_OptimizationProblem}
Having defined the system model notations in Section~\ref{sec:sys}, we proceed to formulate the joint aggregation tree construction and scheduling problem to find the optimal tree as follows:
\bse
\bee
\textsf{Z}:\quad &&\max_{\psi\in \mathcal{T}(\mathcal{G})}\quad  QoA_{\psi}(\vec{W})\label{eq:Objective} \\
&&\text{s.t.}
\quad \forall i \in \{S\}\cup \mathcal{V}\backslash \mathcal{V}^\psi _{\textrm{leaf}} : \forall C \subseteq \{(j,i):(j,i)\in \xi^\psi\},\nonumber \\
&&\quad\label{eq:constraint_DI} \sum_{\scriptsize j:(j,i)\in C}{n_j^\psi \leq  W^\psi _i - M_i^{\psi},}\\
&&\quad \quad W^\psi _i\in \{0,1,\textellipsis , D-1\},\quad \forall i\in \mathcal{V},\label{eq:constrsint_W}\\
&&\quad \quad  W^\psi_S=D,\\
&&\quad \quad  n_i^\psi\in \{0,1\},\quad \forall i \in \mathcal{V}, \label{eq:constraint_n}
\eee
\ese
where $M_i^{\psi} = \min_{\scriptsize j:(j,i)\in C}W^\psi _j$ is the minimum waiting time of the children of node $i$ on tree $\psi$. Constraints~\eqref{eq:constrsint_W}-\eqref{eq:constraint_n} enforce the feasible set of waiting times according to the definitions. The most important constraint is given by Equation~\eqref{eq:constraint_DI} to ensure that the deadline and interference constraints are not violated.
Constraint~\eqref{eq:constraint_DI} states that in a feasible scheduling, the summation of number of participant children of a parent node $i$ and the minimum waiting time of its children should be less than or equal to node $i$'s waiting time, $W_i^{\psi}$. We explain this constraint in detail. Observe that a selected children of $i$ can only be assigned a waiting time of $W_i^{\psi}-1, \dots , 0$ due to deadline constraint in the parent node. Moreover, no two children of $i$ can have a same waiting time otherwise, their simultaneous transmissions will be interfered in the parent node. Therefore, parent $i$ can choose at most $W_i^{\psi}$ children with distinct assigned waiting times chosen from the set $\{W_i^{\psi}-1, \dots , 0\}$. In addition, the value of $M_i^{\psi}$ is not always zero because in some cases number of selected children of $i$ is less than $W_i^{\psi}$. Therefore, transmissions can only occur in time slots $M_i^{\psi}, M_i^{\psi}+1, \dots ,W_i^{\psi}-1$ which gives a total of $W_i^{\psi}-M_i^{\psi}$ time slots. Since in each time slot we have at most one transmission (interference constraint), the total number of selected children cannot exceeds $W^{\psi}_i-M^{\psi}_i$.
\subsection{NP-hardness}
The problem of finding the optimal tree is hard to solve as the number of trees in the network is extremely large in reality. For example, in a complete network graph with $V$ nodes and a sink, the number of feasible trees is $V^{V-2}$. We prove that problem $Z$ is at least as hard as a variant
of classical Maximum Coverage Problem (MCP) called Maximum Coverage Problem with Group Budget
Constraint (MCPG) which is known to be NP-hard \cite{Chekuri}.

\textbf{ Maximum Coverage Problem (MCP). }Given a collection of $n$ sets $U=\{S_1, S_2,\textellipsis , S_n\}$ and a number $l$, the 
goal of MCP is to form set $U^\prime$ by choosing at most $l$ sets from $U$ 
such that the union of selected sets has the maximum cardinality:
\begin{align}
&MCP: \quad \max_{U^\prime} \left| \bigcup_{S_i \in U^\prime} S_i \right|, \quad
 \quad \textrm{s.t.}\quad U^\prime \subseteq U,
\quad \left| U^\prime \right| \leq l. \nonumber
\end{align}

\textbf{ Maximum Coverage Problem with Group Budget Constraint (MCPG). }In \cite{Chekuri}, MCPG is introduced as a general case of MCP.
In MCPG, $n$ sets $S_1,\textellipsis , S_n$ at MCP are partitioned 
to $L$ groups $G_1,\textellipsis , G_L$. MCPG has two 
versions namely cost and cardinality versions where 
the latter is our interest. In the cardinality version of MCPG, given number $l$,
we should select at most $l$ sets from $U$ such that the cardinality of union of the
selected sets is maximized. Moreover, we are permitted to choose at most one set of each group. MCPG is clearly NP-hard because MCP which is known to be NP-hard \cite{Chekuri} is a special case of MCPG where 
each set in $U$ is considered as a group. 
\begin{align}
MCPG:\ \max_{U^\prime} &\left| \bigcup_{S_i \in U^\prime} S_i \right|\\
\textrm{s.t.}\quad &U^\prime \subseteq U, \nonumber\\
&\left| U^\prime \right| \leq l,\nonumber\\
&\left| U^\prime \cap G_i \right| \leq 1, \forall i\in\{1,\textellipsis , L\}. \nonumber
\end{align}
The similarity between our tree construction problem and MCPG is that in both cases 
the objective is to maximize the cardinality. In MCPG we can choose at most one set 
from each group.  Similarly, in problem $Z$, each node can subscribe (cover) different set of sensor nodes based on its deadline and we are allowed to choose at most one set according to the assigned deadline.

\begin{thm}
\label{thm_NP}
Problem $Z$ is NP-hard.
\end{thm}

\begin{proof}
To prove, we reduce MCPG to problem $Z$ with a polynomial time algorithm.
To this end, we construct network graph $\mathcal{G}$ such that the sink is
directly connected to $L$ non-source sensor nodes $C_1, \dots , C_L$ where $L$ is the number of groups in MCPG. There are $V$ other sensor nodes all considered as source nodes connected to $C_1, \dots , C_L$ either directly or indirectly where $V$ is equal to the total number of distinct elements in all groups. That is, $V=\sum_{i=1}^{L} {\sum_{j=1}^{|G_i|} {|g_{i,j}|}}$ where $|g_{ij}|$ is the cardinality of $j^{th}$ set in group $i$ and $|G_i|$ is the number of sets in group $i$. 
Then, we set the sink deadline to $D\geq N$ where $N$ is the total number of sets in $L$ groups, i.e., $N=\sum_{i=1}^L\left| G_i\right|$.
We connect $V$ sensor nodes to $C_1, \dots , C_L$ and to each other such that if we assign a deadline of $D-((\sum_{k=1}^{i-1} {|G_k|})+j-1)$ to the sink's neighbor $C_i$,
$j^{th}$ set of $G_i, 1\leq j\leq |G_i|$ denotes the maximum cardinality set of the sensor nodes who will participate in data aggregation as the successors of $C_i$ in a sub-tree rooted at this node in aggregation tree.
An optimal assignment of deadlines to $C_1, \dots , C_L$ is equal to select at most one set from
each group of MCPG where this optimal assignment results in maximizing both the number of participants in
data aggregation tree as well as the number of covered elements in MCPG. Therefore, a polynomial time optimal algorithm of problem $Z$ leads to a polynomial solution of MCPG which completes the proof.
\end{proof}

\section{Markov-Based Approximate Solution}
\label{sec:markov}
Since problem $Z$ is NP-hard, it is not possible to devise a computationally-efficient algorithm for the optimal solution even in a centralized manner. As such, we pursue approximate solutions. Among different approximation methods, we leverage Markov approximation framework \cite{Chen2} to propose an efficient near-optimal solution for the problem. Generally, in this framework the goal is to tackle combinatorial optimization problems in distributed manner so as 1) to construct a class of problem-specific Markov chains with a target steady-state distribution and 2) to investigate a particular structure of Markov chain that is amenable to distributed implementation. 
We first begin with a brief primer of the theoretical approximation framework \cite{Chen2} in the next subsection. 

\subsection{Markov Approximation}
Recall that $\mathcal{T}$ denotes the set of all possible trees (configurations) of the network. For notational convenience, let us define ${\Phi_\psi ^D=\max (QoA_{\psi}(\vec{W}))}$ , i.e., when the network relies on tree $\psi\in \mathcal{T}$ for data aggregation and sink deadline is $D$, the maximum data aggregation quality is $\Phi_\psi ^D$. We note that $\Phi_\psi ^D$ is obtained by running the optimal scheduling algorithm subject to all the constraints in problem $Z$. In other words, all possible trees are \textit{feasible} configurations in our tree construction algorithm and the scheduling policy ensures the feasibility in terms of respecting the interference and deadline constraints in Equation~\eqref{eq:constraint_DI}. Denote $p_{\psi}$ as the percentage of time that configuration $\psi$ is employed to accomplish data aggregation. 
Using these notations we can rewrite problem $Z$ as follows:
 \begin{align*}
 Z^{\mathbf{eq}}: \max_{\{p_{\psi} \geq 0, \psi \in \mathcal{T}\}} \quad \sum_{\psi \in \mathcal{T}} p_{\psi} \Phi_\psi ^D,\quad
  \textrm{s.t.} \quad \sum_{\psi \in \mathcal{T}} p_{\psi} = 1.
 \end{align*}
To derive a closed-form of the optimal solution of problem $Z^{\mathbf{eq}}$ and to open new design space for exploring a distributed algorithm, we formulate problem $Z^\beta$ as an approximate version of $Z^{\mathbf{eq}}$ using \textit{log-sum-exp} approximation \cite{Chen2}
\begin{align}
\label{eq:z-beta}
Z^\beta:\max_{\{p_{\psi} \geq 0, \psi \in \mathcal{T}\}} &\sum_{\psi\in\mathcal{T}}p_\psi \Phi_\psi ^D - \frac{1}{\beta}\sum_{\psi\in\mathcal{T}}p_\psi\log p_\psi
\\ \textrm{s.t.}\quad \quad \quad &\sum_{\psi\in\mathcal{T}}p_\psi=1,
\end{align}
 where $\beta$ is a large enough positive constant that controls the accuracy of the approximation. Problem $Z^\beta$ is an approximated version of problem $Z$ off by an entropy term $-\frac{1}{\beta}\sum_{\psi\in\mathcal{T}}p_\psi\log p_\psi$ and it is a convex optimization problem so by solving KKT conditions, its optimal solution is obtained by 
\begin{align}
\label{eq:p_star}
p_\psi^\ast =\frac{\exp\big(\beta \Phi_\psi ^D\big)}{\sum_{\psi^\prime\in\mathcal{T}} \exp\big(\beta \Phi_{\psi^\prime} ^D\big)},\quad \psi\in\mathcal{T}.
\end{align}
Moreover, the optimal value is 
\begin{align}
\label{eq:opt_obj}
\widehat{\Phi}_\psi ^D=-\frac{1}{\beta} \log \Bigg( \sum_{\psi\in\mathcal{T}} \exp\big(\beta \Phi_{\psi} ^D\big)\Bigg).
\end{align}
Finally, the approximation gap is characterized as:
\begin{align}
\label{eq:opt_gap}
|\max_{\psi\in \mathcal{T}} \Phi_\psi ^D - \widehat{\Phi}_\psi ^D|\leq \frac{1}{\beta}\log |\mathcal{T}|,
\end{align}
where the approximation gap approaches to zero as $\beta$ approaches to infinity. This means that with larger values of $\beta$ the approximation model is more accurate. 

 In the next step, our endeavor is to obtain the solution of problem $Z^\beta$ by time-sharing among different tree configurations according to $p_\psi^\ast$ in Eq. \eqref{eq:p_star}. According to the basic framework, the key is to investigate a well-structured and distributed-friendly Markov chain whose stationary distribution is $p_\psi^\ast$.

\subsection{Markov Chain Design}
We design a time-reversible Markov chain with states space being $\mathcal{T}$ and the stationary distribution being $p_\psi^\ast$. 
Then, we use this Markov chain structure to hop (migrate) among different states (trees) such that a tree with high $QoA$ has more chances to be visited by Markov random walks. The problem is solved when the Markov chain converges to the ideal steady-state distribution.

Given the Markov chain state space, the next step is to construct the transition rate between two states. Let $\psi,\psi' \in \mathcal{T}$ be two states of Markov chain and $q_{\psi,\psi'}$ be the transition rate from $\psi$ to $\psi'$. 
Herein, the theoretical framework enriches us by two degrees of freedom. It turns out that the key in designing distributed algorithms is to design a Markov chain such that (i) the Markov chain is irreducible (i.e., any two states are reachable from each other) and (ii) the detailed balance equation is satisfied (i.e., $p_{\psi}^\ast q_{\psi,\psi^\prime}=p_{\psi^\prime}^\ast q_{\psi^\prime,\psi}, \forall \psi,\psi^\prime\in\mathcal{T}$). Consequently, we are allowed to set the transition rates between any two states to be zero if they are still reachable from any other states. We refer the reader to \cite{Chen2} for further explanation.

In practice, however, direct transition between two states means migration between two tree structures. To derive a distributed algorithm, we only allow direct transitions between two states if the current and the target trees can be transformed to each other by only one parent changing operation in one of the trees. Namely, two states $\psi$ and $\psi^\prime$ are directly reachable from each other if we can construct tree $\psi^\prime$ by deleting
an edge $(i,j)\in\xi$ from $\psi$ and adding edge $(i,k)\in\xi$ to $\psi$. 
Using this transition structure the next step is to set the transition rate as follows: 
\begin{align}
\label{eq:trans_rate}
q_{\psi,\psi^\prime}=\frac{1}{\exp (\alpha)}\frac{\exp (\beta\Phi_{\psi'}^D)}{\exp (\beta\Phi_{\psi}^D)+\exp (\beta\Phi_{\psi'}^D)}
\end{align}
where $\alpha \geq 0$ is a constant and $q_{\psi^\prime ,\psi}$ is defined symmetrically.

\subsection{Algorithm Design}
Our goal is to realize a distributed implementation of the Markov chain proposed in the previous section. In this part, we detail our implementation.

To compute transition rate between the states, the maximum $QoA$ of both the current ($\Phi_\psi ^D$) and the target ($\Phi_{\psi^\prime}^D$) states are required. To calculate these values we employ the scheduling algorithm proposed in [10], namely ``Waiting-Assignment''  algorithm. ``Waiting-Assignment''  is a distributed polynomial time algorithm  to find the optimal waiting time of the nodes and hence $\Phi_\psi ^D$ in  tree $\psi$.

Our algorithm runs as follows.
Given initial aggregation tree $\psi$ and deadline $D$, we first run ``Waiting-Assignment'' algorithm to obtain $\Phi_\psi ^D$. Then, based on the underlying Markov chain design and in an iterative manner, we proceed to migrate to a target aggregation tree $\psi^\prime$ with (probably) better $\Phi_{\psi^\prime}^D$ than $\Phi_{\psi}^D$.
To realize this end, each sensor node individually runs ``Parent-Changing'' algorithm which is summarized as Algorithm 1. 
%

\vspace{2mm}
\begin{algorithm}
\caption{``Parent-Changing'' algorithm for node $i\in\mathcal{V}$}
\label{Parent-Changing}
 \DontPrintSemicolon 
\KwIn{$\alpha ,\beta$}
\KwOut{New parent of node $i$}
\BlankLine

$P_i \leftarrow$ parent of node $i$

$\mathcal{N}_{\geq i} \leftarrow \{j: (i,j)\in \xi, W_j\geq W_i\}$ 

Node $i$ generates a timer $\tau_i\thicksim\exp (\lambda _i)$ with mean 
$\lambda_i =\frac{1}{|\mathcal{N}_{\geq i}|}$ and starts to count down

When $\tau_i$ expires, node $i$ randomly selects one of its neighbors
$P_i^\prime\in \mathcal{N}_{\geq i}$.

$\Phi _{\textrm{prev}} \leftarrow $ node $i$'s estimation of $\Phi_{\psi}^D$ in Equation (\ref{eq:trans_rate}), i.e., the maximum $QoA$ of the  current tree

Node $i$ changes its parent to $P_i^\prime$

$\Phi _{\textrm{next}} \leftarrow $ node $i$'s estimation of $\Phi_{\psi'}^D$ in Equation (\ref{eq:trans_rate}), i.e., the maximum $QoA$ of the new tree



With probability $q_{\psi ,\psi'}$, node $i$ keeps the new tree configuration and with probability $1-q_{\psi ,\psi'}$ switches back and connects to the previous parent $P_i$

\If {$i$ changed its parent in Step 8}
{

$P_i^\prime$ invokes ``Waiting-Assignment'' algorithm on its sub-tree

$P_i$ invokes ``Waiting-Assignment'' algorithm on its sub-tree
}

Node $i$ refreshes the timer and begins counting down 
\end{algorithm}
The detailed description of Algorithm 1 is as follows. In Line 3, an exponentially distributed random number with mean $\lambda_i =\frac{1}{|\mathcal{N}_{\geq i}|}$ is generated as the timer value in which this setting is required to ensure the convergence of the corresponding Markov chain. In Line 4,
node $i$ selects a new parent $P^\prime _i$ such that $W_{P^\prime _i}\geq W_i$. This ensures that after the parent changing, the data structure still remains a tree since the new structure is not a tree only if node $i$ chooses its new parent from its successors where all have a less waiting time than node $i$'s waiting time.
Meanwhile, this strategy is also rational because finding a new parent with a shorter waiting time declines node $i$'s new waiting time which probably reduces $QoA$. In Lines 5-7, node $i$ temporarily changes its parent and estimates the impact of this change on the maximum $QoA$ of data aggregation. Based on the estimation and transition rate given by Equation (\ref{eq:trans_rate}), in Line 8, node $i$ decides whether to keep its new parent or not.  If the new state is established, then nodes $P_i$ and $P_i^\prime$ should run ``Waiting-Assignment'' algorithm to update waiting time of their successors because of their sub-tree changes. It is worthy to note that the parameter $\beta$ not only affects the accuracy of the approximation, but also with large values of $\beta$, the algorithm migrates towards better configurations more greedily, whereas it may lead to premature convergence and trap into local optimum trees.
\begin{propos}
``Parent-Changing'' algorithm in fact implements a time reversible Markov chain
with stationary distribution in Equation (\ref{eq:p_star}).
\end{propos}
\begin{proof}
The designed Markov chain is finite space state ergodic Markov chain where each tree configuration in state space is reachable from any other state by one or more 
parent changing process. We proceed to prove that the stationary state of designed 
Markov chain is Equation (\ref{eq:p_star}).
Let $\psi\rightarrow\psi'$ denote transition from state $\psi$ to $\psi'$ at a timer expiration
and $A=\frac{1}{\exp (\alpha)}\frac{\exp (\beta\Phi_{\psi'}^D)}{\exp (\beta\Phi_{\psi}^D)+\exp (\beta\Phi_{\psi'}^D)}$.
Moreover, $\textrm{Pr}(\psi\rightarrow\psi')$ is the probability of this transition.

This probability can be calculated as follows:
\begin{align}
\textrm{Pr}(\psi\rightarrow\psi')&=\textrm{Pr}(i\textrm{\ chooses } P'|i\textrm{'s timer expires}).\textrm{Pr}(i\textrm{'s timer expires})\nonumber\\ 
&=\frac{1}{|\mathcal{N}_{\geq i}|}.A
.\frac{|\mathcal{N}_{\geq i}|}{\sum_{j\in \mathcal{V}}|\mathcal{N}_{\geq j}|}=\frac{1}{\sum_{j\in \mathcal{V}}|\mathcal{N}_{\geq j}|}.A
\end{align}
In the algorithm, node $i$ counts down with rate $|\mathcal{N}_{\geq i}|$. Therefore, the rate of leaving state $\psi$ is 
${
\sum_{j\in\mathcal{V}}|\mathcal{N}_{\geq j}|.
}$
We can calculate transition rate $q_{\psi ,\psi'}$ as follows:
\begin{align}
q_{\psi ,\psi'}
=\sum_{j\in \mathcal{V}}|\mathcal{N}_{\geq j}|.\frac{1}{\sum_{j\in \mathcal{V}}|\mathcal{N}_{\geq j}|}.A=A\
\vspace{-3mm}
\end{align}
We can see that $p^\ast _\psi.q_{\psi ,\psi'}=p^\ast _{\psi'}.q_{\psi' ,\psi}$. Therefore, the detailed balance equation holds and the stationary distribution of constructed Markov chain is Equation (\ref{eq:p_star}) \cite{Zhang}.
\end{proof}

``Parent-Changing'' algorithm is distributed if we can estimate $\Phi _{\textrm{next}}$ and $\Phi _{\textrm{prev}}$ in the algorithm in a distributed manner.
By exact calculation of these values, the designed Markov chain will converge to 
stationary distribution in Equation (\ref{eq:p_star}). Hence, ``Parent-Changing'' algorithm can give us
a near-optimal solution of problem $Z$. However, exact calculation of $\Phi _{\textrm{next}}$ and $\Phi_{\textrm{prev}}$ is not possible in nodes locally since they can only be calculated in the sink by running ``Waiting-Assignment'' algorithm. Therefore, we need to estimate their values. We estimate the values by two different methods.

\textbf{ Approx-1: First method of estimating $\Phi _{\textrm{next}}$ and $\Phi _{\textrm{prev}}$.} When node $i$ wants to modify its parent from $P_i$ to $P_i^\prime$ (and subsequently tree $\psi$ to $\psi'$), one possible way of estimation is running ``Waiting-Assignment'' algorithm by nodes $P_i$ and $P_i^\prime$ on their sub-trees.
Let $\Phi_{\textrm{prev}}[s]$ and $\Phi_{\textrm{next}}[s]$ denote the maximum achievable $QoA$s in a sub-tree rooted at node $s$ before and after the sub-tree change, respectively. Then, we have the following estimation: 
\begin{align}
\label{eq:QoAestimate1}
\Phi _{\textrm{next}}\thickapprox (\Phi_{\textrm{next}}[P_i]+\Phi_{\textrm{next}}[P'_i]),\ \ \Phi _{\textrm{prev}}\thickapprox (\Phi_{\textrm{prev}}[P_i]+\Phi_{\textrm{prev}}[P'_i]).
\end{align} 

When node $i$ changes its parent from $P_i$ to $P_i^\prime$,
only sub-trees rooted at $P_i$ and $P_i^\prime$ change and all other parts of the tree
remain intact and so the estimation accuracy is expected to be high. This estimation comes with the overhead of running ``Waiting-Assignment'' algorithm at nodes $P_i$ and $P_i^\prime$ to calculate $\Phi_{\textrm{next}}$ and $\Phi_{\textrm{prev}}$.

\textbf{Approx-2: Second method of estimating $\Phi _{\textrm{next}}$ and $\Phi _{\textrm{prev}}$.} Another way of estimation is just using waiting times of nodes $P_i^\prime$ and $i$: 
\begin{align}
\label{eq:QoAestimate3}
&\Phi _{\textrm{next}}\thickapprox W_{P_i^\prime},\ \ \Phi _{\textrm{prev}}\thickapprox W_{i}.
\vspace{-3mm}
\end{align}

A larger value of $W_{P_i^\prime}$ indicates that node $i$ \emph{probably} will be 
assigned a greater waiting time if it joins to sub-tree of $P_i^\prime$ and vice versa.
In Section VII, we evaluate the efficiency of both mentioned methods by simulation. 
%
%
%

\subsection{Perturbation Analysis}
In ``Parent-Changing'' algorithm, if we obtain the accurate value of $\Phi_\psi ^D$ to calculate transition rates, the designed Markov chain converges to the stationary distribution given by Equation (\ref{eq:p_star}). Thus, we have a near-optimal solution of problem $Z$ with optimality gap determined in Equation (\ref{eq:opt_gap}).
In distributed fashion, however, we estimate the optimal tree-specific $QoA$s by Equations ~\eqref{eq:QoAestimate1} and \eqref{eq:QoAestimate3}. 
Consequently, the designed Markov chain may not converge to the stationary distribution in Equation (\ref{eq:p_star}). Fortunately, our employed theoretical approach can provide a bound on the optimality gap due to the perturbation errors of the inaccurate estimation using a quantization error model.

We assume that in a tree configuration $\psi$, the corresponding perturbation error is bounded to $[-\Delta _{\psi},\Delta _{\psi}]$. In order to simplify the approach, we further assume that $\Phi_\psi ^D$ takes only one of the following $2n_{\psi}+1$ values:

\begin{align}
 [\Phi_\psi ^D-\Delta _{\psi},\dots ,\Phi_\psi ^D
-\frac{1}{n_{\psi}}\Delta _{\psi},\Phi_\psi ^D, \Phi_\psi ^D
+\frac{1}{n_{\psi}}\Delta _{\psi},\dots ,\Phi_\psi ^D+\Delta _{\psi}],
\end{align}

where $n_{\psi}$ is a positive constant. Moreover, with probability 
$\eta_{j,\psi}$, the maximum quality of aggregation is 
equal to $\Phi_\psi ^D+\frac{j}{n_{\psi}}\Delta _{\psi}, \forall j\in\{-n_{\psi},\dots ,n_{\psi}\}$ 
and $\sum_{j=-n_{\psi}}^{n_{\psi}} \eta_{j,\psi}=1$.


Let $\tilde{p}$ denote the stationary distribution of the \emph{states} in the \textit{perturbed} Markov chain \cite{Zhang}. We also denote stationary distribution of the \emph{configurations} in the original and perturbed Markov chains by $p^\ast:\{p^\ast _\psi ,\psi\in\mathcal{T}\}$ and $\bar{p}:\{\bar{p}_\psi ,\psi\in\mathcal{T}\}$, respectively. Then, we have \cite{Zhang}
\begin{align}
\tilde{p}\triangleq[\tilde{p}_{\psi,\Phi_\psi ^D +\frac{j}{n_\psi}\Delta _\psi},j\in\{-n_{\psi},\dots ,n_{\psi}\},\psi\in\mathcal{T}],\\
\bar{p}_\psi (\Phi)=\sum_{j\in\{-n_{\psi},\dots ,n_{\psi}\}} \tilde{p}_{\psi,\Phi_\psi ^D +\frac{j}{n_\psi}\Delta _\psi},\forall \psi\in\mathcal{T}.
\end{align}
Using total variance distance \cite{Diaconis} we can measure the distance of $p^\ast _\psi$ and $\bar{p}_\psi$ as 
\begin{align}
d_{TV}(p^\ast ,\bar{p})\triangleq \frac{1}{2}\sum_{\psi\in\mathcal{T}}|p^\ast _\psi -\bar{p}_\psi|.
\end{align}
\begin{thm}
\label{thm:perturb}
\emph{a)} The total variance distance between $p^\ast _\psi$ and $\bar{p}_\psi$ is bounded by $[0,1-\exp (-2\beta \Delta _{\textrm{max}})]$ where $\Delta _{\textrm{max}}=\max _{\psi\in\mathcal{T}}\Delta _\psi$.
\emph{b)} By defining $\Phi_{\max} = \max_{\psi \in \mathcal{T}} \Phi_\psi^D$, the optimality gap in $|p^\ast - \bar{p}|$ is 
\begin{align}
\label{eq:pret-bound}
|p^\ast - \bar{p}|\leq 2\Phi_{\max}(1-\exp (-2\beta\Delta _{\textrm{max}})).
\end{align}
\end{thm}
For proof and remarks, we refer to \cite{Zhang}.

Finally, we highlight that the convergence (mixing time) of the algorithm based on Markov approximation framework is studied in \cite{Zhang} and \cite{Shao2013}. Overall, this framework, in its basic setting suffers from slow rate of convergence. To mitigate this, a promising solution is to find a ``good'' initial aggregation tree as the input for the iterative Alg.~\ref{Parent-Changing}, thereby its convergence time can be improved. In this regard, the goal in the next section is to develop an algorithm to construct a fast initial aggregation tree.

\section{Initial Tree Construction Algorithm}
\label{sec:init}

We proceed to develop an initial tree construction algorithm to identify a good starting feasible solution for bootstrapping the Markov approximation-based algorithms. The intuition is that if Alg.~\ref{Parent-Changing} starts from a tree which already has a good quality, not only high-quality data aggregation experience can be provided starting from the beginning, but also fast convergence of the algorithm proposed in the previous section can be achieved. Alg. \ref{InitialTree} outputs a near optimal feasible aggregation tree which is a spanning tree over the underlying WSN topology.

\vspace{4mm}
\begin{algorithm}
\caption{``FastInitTree''}
\label{InitialTree}
\DontPrintSemicolon
\KwIn{Graph $\mathcal{G}=\{\mathcal{V}\cup \{S\},\mathcal{E}\}$, Deadline $D$}
\KwOut{Close-to-optimal spanning tree}
\BlankLine

Define $\mathcal{V}_{\textsf{done}}$ as the set of nodes that their parent in final tree is identified

$\mathcal{V}_{\textsf{done}}\leftarrow \emptyset$

\textbf{ExtendTree}($\mathcal{G}$, $S$, $D$)

For any node not in $\mathcal{V}_{\textsf{done}}$ assign it to one of its neighbors in $\mathcal{V}_{\textsf{done}}$ with the least number of children
\end{algorithm}

\vspace{6mm}

\begin{algorithm}
\caption{``Extend-Tree''} \label{ExtendTree}
\DontPrintSemicolon
\KwIn{Graph $\mathcal{G}=\{\mathcal{V}\cup \{S\},\mathcal{E}\}$, Parent $P$, Deadline $D$}
\KwOut{A tree rooted at parent $P$}
\BlankLine

Define $\mathcal{V}_{\textsf{done}}$ as set of nodes that assigned to a parent (initially, $\mathcal{V}_{\textsf{done}}=\emptyset$)

$\mathcal{V}_{\textsf{done}}\leftarrow \mathcal{V}_{\textsf{done}} \cup P$

$\mathcal{V}_{\textsf{curr}}\leftarrow$ neighbors of $P$ except those in $\mathcal{V}_{\textsf{done}}$

$k\leftarrow |\mathcal{V}_{\textsf{curr}}|$ 

$power_i\leftarrow$ number of neighbors of $i$ in $\mathcal{V}\backslash \mathcal{V}_{\textsf{done}}$

w.l.g assume that $c_1,\dots ,c_{k}$ are the members of $\mathcal{V}_{\textsf{curr}}$ sorted in a descending order based on $power_i, i=1,\dots ,k$

\For{i=1:min\{k,D\}}{

Set $P$ as parent of $c_i$ in the output tree

\If{$D>0$}{

\textbf{ExtendTree}($\mathcal{G},c_i,D-i$)
}
}
\end{algorithm}
In a nutshell, Alg.~\ref{InitialTree} aims to find an appropriate unique parent for each node (except the sink) in the graph. 
By doing so, a feasible aggregation tree is indeed constructed.
 In particular, $\mathcal{V}_{\textsf{done}}$ includes the nodes that their parents are chosen and initially defined as an empty set. Alg.~\ref{InitialTree} calls Alg. \ref{ExtendTree} on the sink node and receives a tree rooted at the sink. However, the returned tree by Alg. \ref{ExtendTree} may not be a spanning tree, i.e., some nodes may not be still in $\mathcal{V}_{\textsf{done}}$. To construct the spanning tree, Alg. \ref{InitialTree} goes through the nodes that are not in $\mathcal{V}_{\textsf{done}}$ and set them as children of one of their neighbors in $\mathcal{V}_{\textsf{done}}$ with the least number of children (in Line 4 of \emph{FastInitTree}). The intuition behind this is that a parent with less children is more likely to be able to participate its new child in the aggregation. Before proceeding to explain the details of Alg.~\ref{ExtendTree}, we give definition of ``\textit{well-structured}'' graph in the context of deadline-constraint aggregation tree construction, which is useful in the discussion of the algorithm. 
\begin{defi1}
	\label{def:1}
	Graph $\mathcal{G}=\{\mathcal{V}\cup \{S\},\mathcal{E}\}$ with $V=|\mathcal{V}|$ is well-structured under a specific sink deadline $D$ if the optimal tree in $\mathcal{G}$ is an optimal tree in a complete graph with $V$ nodes where $V\geq 2^D-1$.
\end{defi1}

\begin{figure*}[t!]
	\centering
	\vspace{-0.05em}
	\begin{subfigure}[b]{0.25\textwidth}
		\begin{center}
			\includegraphics[width=\textwidth]{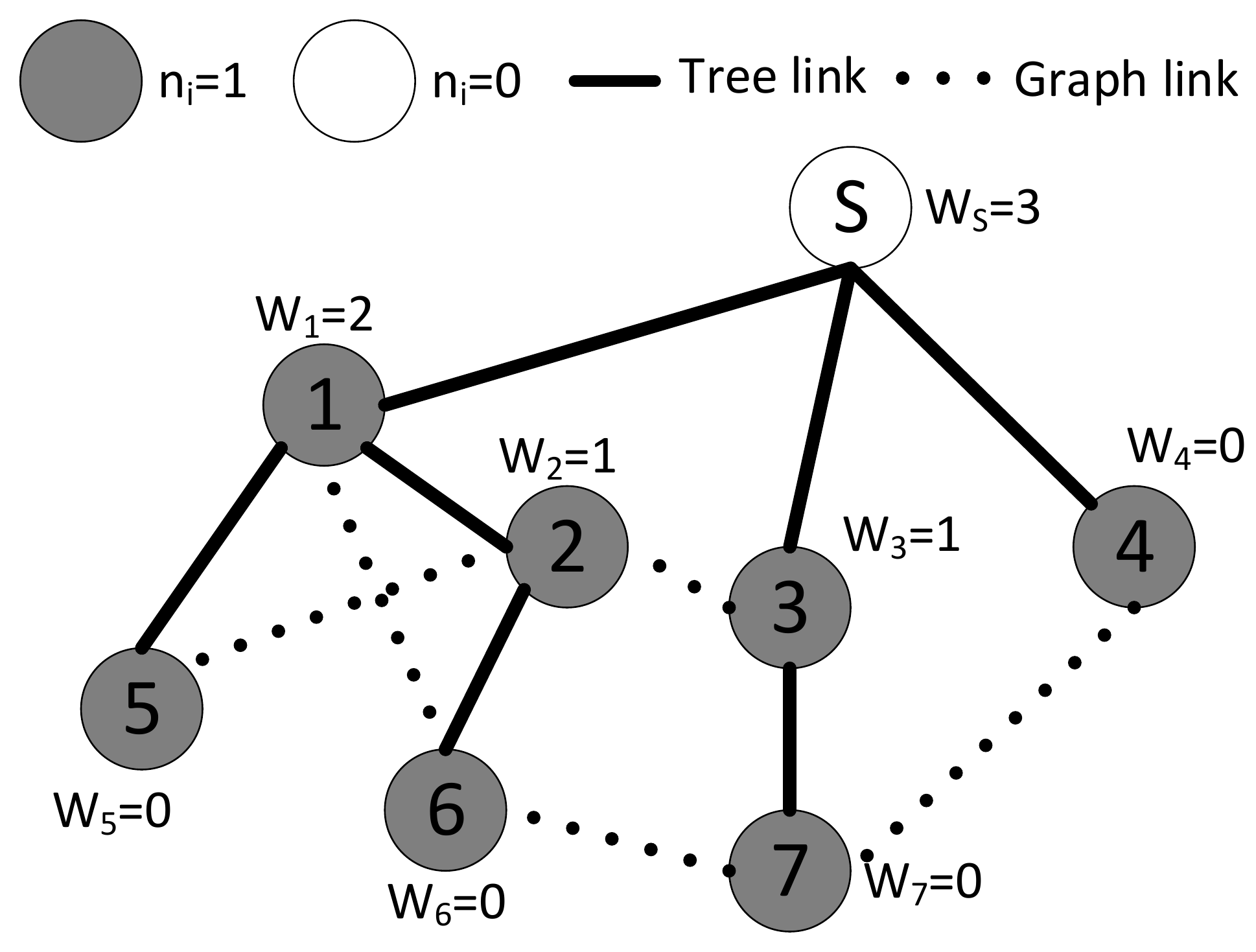}
			\caption{A \textit{well-structured} graph with $12$ links.}
			\label{fig:Ex3a}
		\end{center}
	\end{subfigure}%
	\hspace{20mm}
	~ 
	\begin{subfigure}[b]{0.25\textwidth}
		\begin{center}
			\includegraphics[width=\textwidth]{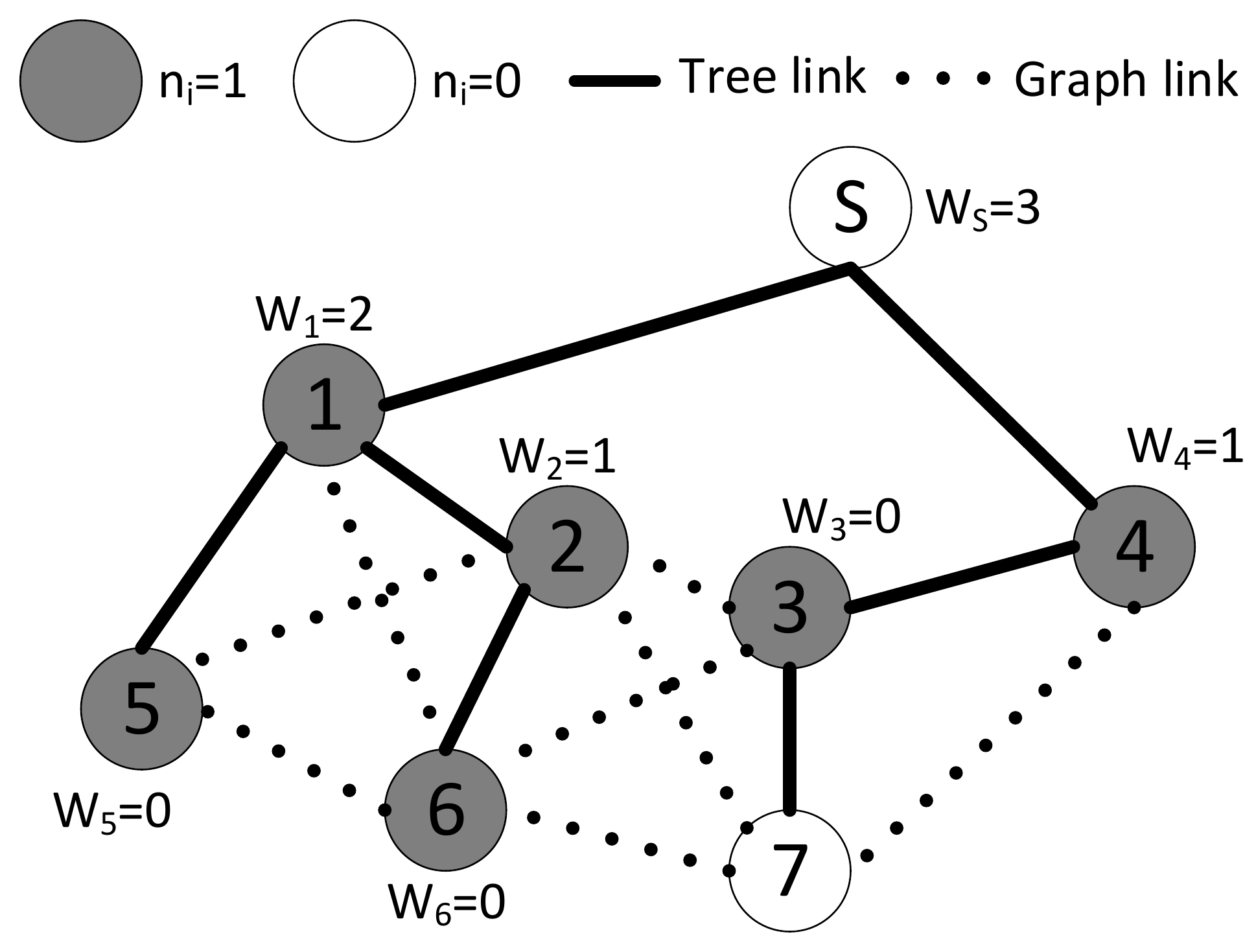}
			\caption{A graph that is not \textit{well-structured} with $15$ links. }
			\label{fig:Ex3b}
		\end{center}
	\end{subfigure}%
	\vspace{2mm}
	\caption{Illustration of \textit{well-structured} graph with sink deadline $D=3$. Although the graph in Fig.~\ref{fig:Ex3b} is more rich in terms of number of  links, the structure of graph does not allow to participate all nodes under any feasible tree.}
	\label{fig:Ex3}
\end{figure*}

Alg. \ref{ExtendTree} is the main part of the \emph{FastInitTree} algorithm. It works recursively on the input parent $P$ to develop a tree. Note that the maximum number of participant nodes with deadline $D$ is $2^D-1$ (see \cite{Alinia} for the proof) that is achievable if the graph is \emph{well-structured} for deadline $D$, thereby its structure allows to construct a tree with the maximum feasible $QoA$ of $2^D-1$. We emphasize that if a graph is well-structured, it does not imply that the number of communication links in the graph is equal (or even close) to the number of links in the corresponding complete graph. For further illustration, let's consider the well-structured graph in Fig.~\ref{fig:Ex3a} with $7$ (i.e., $2^3-1$) nodes and $12$ links. 
Despite more connectivity among the nodes in the graph of Fig.~\ref{fig:Ex3b}, it is not a well-structured graph and its optimal $QoA$ is less than the optimal $QoA$ in Fig.~\ref{fig:Ex3a}.
Finally, we call the optimal tree of a well-structured graph $\mathcal{G}=\{\mathcal{V}\cup \{S\},\mathcal{E}\}$ as \emph{ideal} tree for deadline $D$ while its maximum $QoA$ is $2^D-1$.
  
Now, we turn back to explain the main idea behind developing  Alg.~\ref{InitialTree}. In Alg.~\ref{InitialTree} we assume that the network graph is well-structured and try to build an aggregation tree such that its structure is as close as possible to the corresponding \textit{ideal} tree. It is not difficult to see that  in \textit{ideal} tree, the number of children of each node (including sink) is equal to its waiting time (as in Fig. \ref{fig:Ex3a}). Based on this fact, Alg.~\ref{InitialTree} starts from the sink node and by calling Alg.~\ref{ExtendTree} tries to find top $D$ most \emph{powerful} neighbors of sink, where the \textit{power} of a node is defined as the number of its neighbors (Line 5 in Alg. \ref{ExtendTree}). Indeed, the algorithm assumes that these $D$ nodes will have waiting times $\{D-1, \dots ,0\}$ according to their ability to communicate with the other nodes, i.e., their power. Then, the algorithm considers these nodes as the sink's children in the final tree. Alg.~\ref{ExtendTree} is called recursively on sink's children to build the rest of the tree. The wisdom of the algorithm in selecting children of each node makes it as a promising method.

\subsection{Discussions on the Optimality and Complexity of Alg.~\ref{InitialTree}}
Theorem \ref{ExtentTreeSpanning} proves that Alg.~\ref{InitialTree} outputs the optimal tree, given that the underlying graph is complete.

\begin{thm}
\label{ExtentTreeSpanning}
Alg.~\ref{InitialTree} generates an optimal aggregation tree given that the underlying graph is a complete graph $\mathcal{G}_k$ with $k\geq 2^D-1$.
\end{thm}

\begin{proof}
Note that in an \emph{ideal} tree for deadline $D$, we have this key property that after running optimal scheduling algorithm \cite{Shroff} on the tree, each node having waiting time $w, 0\leq w\leq D$, has exactly $w$ children with waiting times $\{w-1,\dots ,0\}$. By following the steps of Alg.~\ref{InitialTree}, it can be seen that the algorithm preserves the key property of the ideal tree while constructing the aggregation tree. The recursive tree construction process starts from the sink node. Since we have a complete graph, Alg.~\ref{ExtendTree} is able to choose exactly $D$ out of $N$ nodes with highest \emph{power} as children of the sink. In the next step, Alg.~\ref{ExtendTree} is called on these $D$ children of the sink namely $c_1,c_2,\dots ,c_D$ with $power_i\geq power_{i+1}, i=1,\dots ,D-1$ and chooses $D-i$ children for $c_i, i=\{1,\dots ,D\}$ to follow the property of the ideal tree for deadline $D$. This process, while keeping the key property of the corresponding ideal tree, continues recursively in the same manner on the remaining nodes until all nodes assigned to a parent. Thus, it ensures that the final output is an ideal tree.
\end{proof}

\begin{thm}
	\label{heuristic_complexity}
	The time complexity of Alg.~\ref{InitialTree} is $O(Nv\log v)$ where $v$ is the maximum node degree in the network graph. 
\end{thm}

\begin{proof}
	The cost of Alg.~\ref{ExtendTree} is determined by total sorting cost of children for $O(N)$ nodes which is bounded by $O(Nv\log v)$. Alg.~\ref{InitialTree} goes over nodes who are not in $\mathcal{V}_{\textsf{done}}$ to assign them to a parent which costs $O(N)$. Therefore, the total cost of the Alg. \ref{InitialTree} is $O(Nv\log v)$.
\end{proof}
\subsection{Remarks on the Time-varying Deadline  in the Sink}
Data aggregation is a periodic action in WSNs and in each period, the sink may change aggregation parameters such as deadline. When deadline changes, the previously constructed aggregation tree in the last period may not produce the same level of $QoA$. In this situation, a new tree structure is needed to maximize $QoA$ under the new deadline. A na\"{i}ve approach in this situation is running the tree construction algorithm with the new deadline. However, the following theorem shows that the optimal tree does not need to be changed for the case that the deadline is decreased as compared to its previous value.

\begin{thm}
\label{deadline_reducing}
If all nodes are source and $\psi^{\star}$ is an optimal tree under the sink deadline $D$, then $\psi^{\star}$ is optimal for deadline $D^\prime, D^\prime=1,2,\dots ,D-1$. In addition, the optimal scheduling can be reconstructed by reducing the previous waiting times by $D-D^\prime$.
\end{thm}

\begin{proof}
We prove the theorem when $\psi^{\star}$ is an ideal tree. For the case that $\psi^{\star}$ is not ideal, the proof is similar. With the ideal tree $\psi^{\star}$, there are $2^D-1$ participant nodes in the tree. Each node (including sink) with waiting time $w, w\in \{0,\dots ,D\}$, has exactly $w$ children with assigned waiting time $\{w-1, \dots ,0\}$. We claim that if we keep the same tree for new sink deadline $D^\prime$ with $D^\prime<D$ and reduce the previously assigned waiting times by $D-D^\prime$ then, the new scheduling is feasible and the number of participant nodes is $2^{D^\prime}-1$, i.e., the optimal $QoA$, which in turn proves the optimality of the tree for deadline $D^\prime$. First, the scheduling is feasible since all waiting times are reduced by a constant and so the transmissions occur in the same order as in the previous feasible scheduling.

Second, the number of participant nodes in the new scheduling, namely $X$, can be calculated by subtracting total number of nodes with waiting time less than $D'$ from $2^D-1$ since with the new scheduling, these nodes' waiting times will be negative which has no meaning and makes them non-participant nodes. Therefore, we can sum up all nodes in the previous scheduling having waiting time greater than or equal to $D'$ to find $X$. Formally, we have
$X=f(D-D^\prime)+f(D-D^\prime+1)+\dots +f(D)\nonumber$ where, $f(i)$ denotes the number of nodes in the previous scheduling with assigned waiting time $i$. Note that we have $f(i)=f(i+1)+f(i+2)+\dots +f(D)$ and $f(D)=f(D-1)=1$. Then, we can calculate $\sum_{i=D-D^\prime}^{D} f(i)$ as follow:
\begin{align}
&\overbrace{f(D-D^\prime)}^\text{A}+\overbrace{f(D-D^\prime+1)}^\text{B}+\dots +\overbrace{f(D)}^\text{C}=\nonumber \\
&\overbrace{f(D-D^\prime+1)+f(D-D^\prime+2)+\dots +f(D)}^\text{A}+\overbrace{f(D-D^\prime+2)+\dots +f(D)}^\text{B}+...+\overbrace{f(D)}^\text{C}\nonumber
\end{align}
By solving the above equation we have $X=\sum_{i=D-D^\prime}^{D} f(i)=\sum_{i=0}^{D^\prime -1} 2^i=2^{D^\prime}-1$. 
\end{proof}

Theorem \ref{deadline_reducing} implies that if we construct a near optimal tree for a specific deadline, then the same tree can be used for all shorter deadlines. Most importantly, the new scheduling is straightforward and needs no cost. This can help to avoid the overhead of running the tree construction and scheduling algorithm when deadline changes. 

\section{Simulation Results}

\label{sec:sim}
In this section, we evaluate our proposed algorithms through extensive simulations. Unless otherwise specified, the settings are as follows: 100 sensor nodes uniformly dispersed in a square field with side length of 300m. Sink node is located at the center of the top side of the square field i.e., its position is (150,300). Communication range of nodes is 75m, i.e., two nodes are connected in the network if their distance is ``$\leq 75\text{m}$''. After deployment, sensor nodes construct an initial data aggregation tree. Except for the experiments that use initial tree built by \emph{FastInitTree} algorithm, the tree is constructed based on Greedy Incremental Tree (GIT) algorithm \cite{Krishnamachari2}. 
We let $\alpha =0.2$ and $\beta =2$, and choose 80\% of nodes randomly as the sources. Each data point of the figures belongs to the average value of 50 runs with the 95\% confidence interval where each run is a different random topology. Moreover, for each topology, sink imposes a deadline in terms of time slots uniformly and randomly selected from interval [10,20]. 
We report the results of approximation algorithms after 50 iterations where an iteration is defined as a timer expiration of a sensor node. 
\begin{table}
	\centering
	\caption{Acronyms for the algorithms }
	\label{tb:comparisons}
    \begin{tabular}{ | l| p{10cm} |}
    \hline
    \textbf{Notation} & \textbf{Description} \\ \hline\hline
    $Z$-Optimal & The optimal solution of problem $Z$ implemented using exhaustive search\\ \hline
    Approx-1 & Approximation algorithm that estimates the current $QoA$ using Equation (\ref{eq:QoAestimate1}) in each node (has some overheads)\\ \hline
    Approx-2 & Approximation algorithm that estimates the current $QoA$ using Equation 								(\ref{eq:QoAestimate3}) and in each node (has no overheads) \\ \hline
    Algorithm of \cite{Shroff} & Optimal algorithm presented in \cite{Shroff} to maximize $QoA$
    in a given tree\\ \hline
    Approx-1H & Approx-1 starting from initial tree constructed by FastInitTree\\ \hline 
    Approx-2H & Approx-2 starting from initial tree constructed by FastInitTree\\ \hline
    FastInitTree & Heuristic algorithm\\
    \hline
    \end{tabular}
\end{table}

\begin{figure*}
	\center
		\begin{minipage}[b]{.308\textwidth}%
		\includegraphics[width=\textwidth]{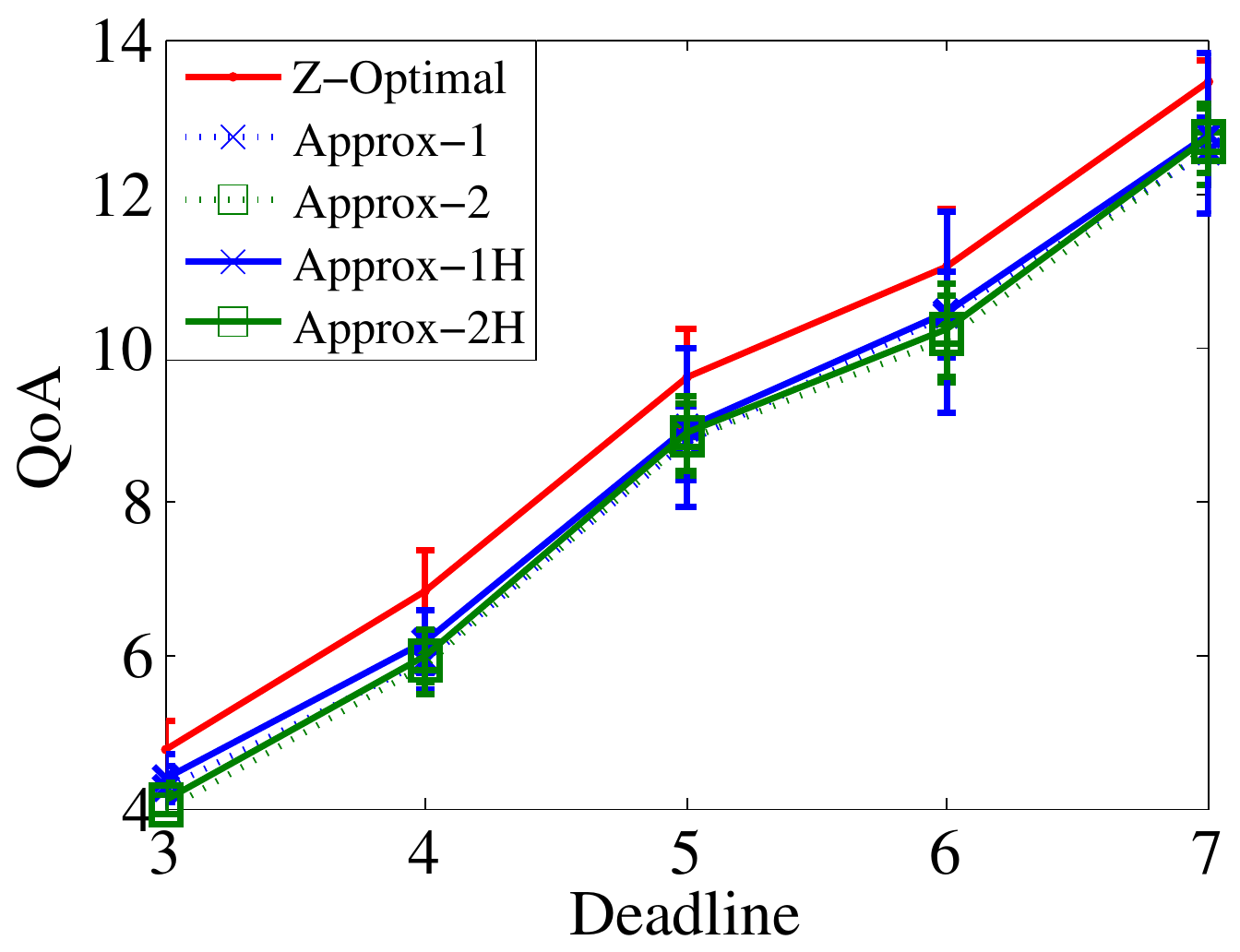}
		\caption{Quality of aggregation vs. deadline ($V=15$).}%
		\label{fig:opt}%
	\end{minipage}
	\hspace{1mm}
	\begin{minipage}[b]{.308\textwidth}%
		\includegraphics[width=\textwidth]{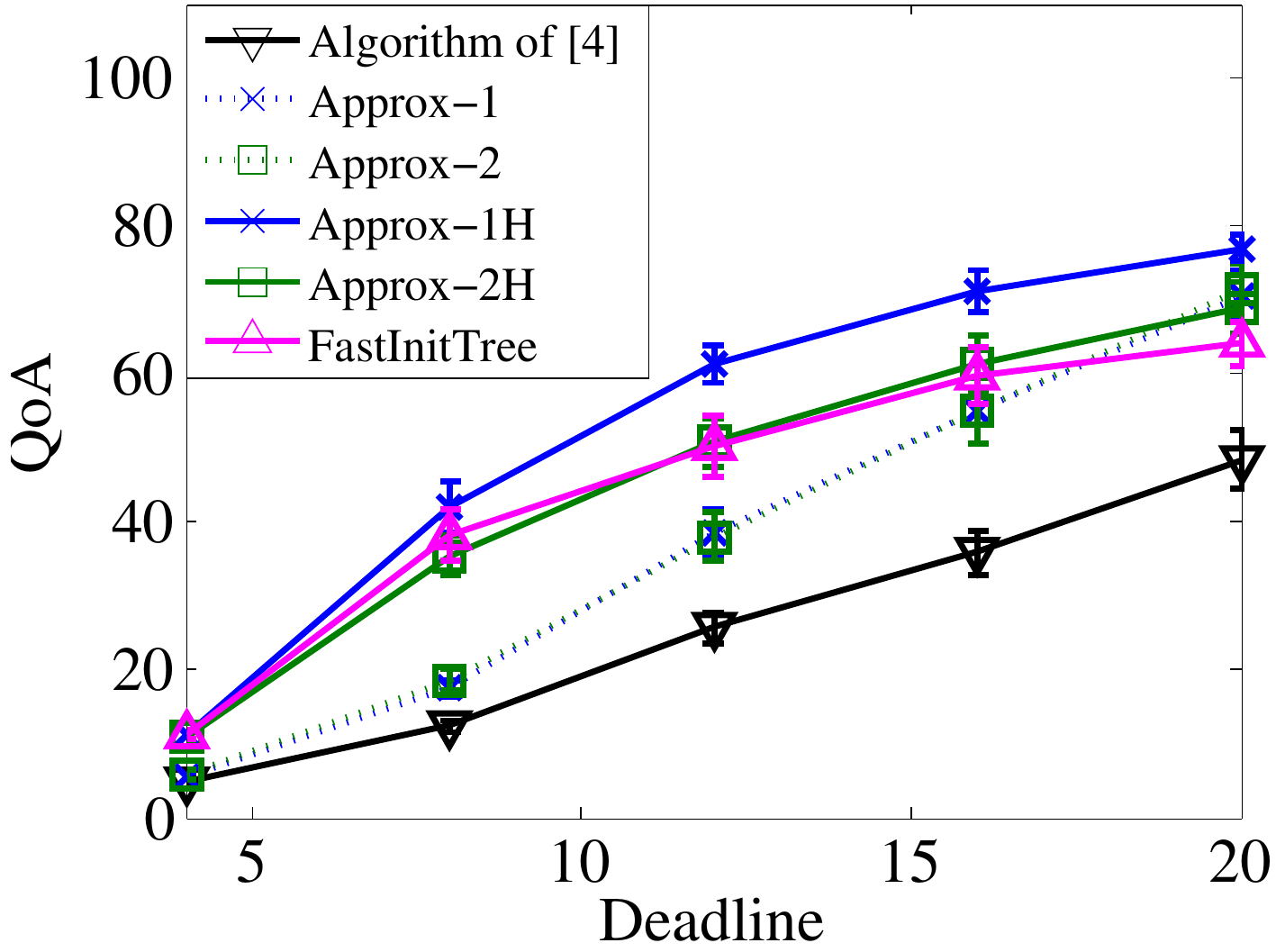}
		\caption{Quality of aggregation vs. deadline ($V=100$).}%
		\label{fig:deadline}
	\end{minipage}
	\hspace{1mm}
	\begin{minipage}[b]{.305\textwidth}%
		\includegraphics[width=\textwidth]{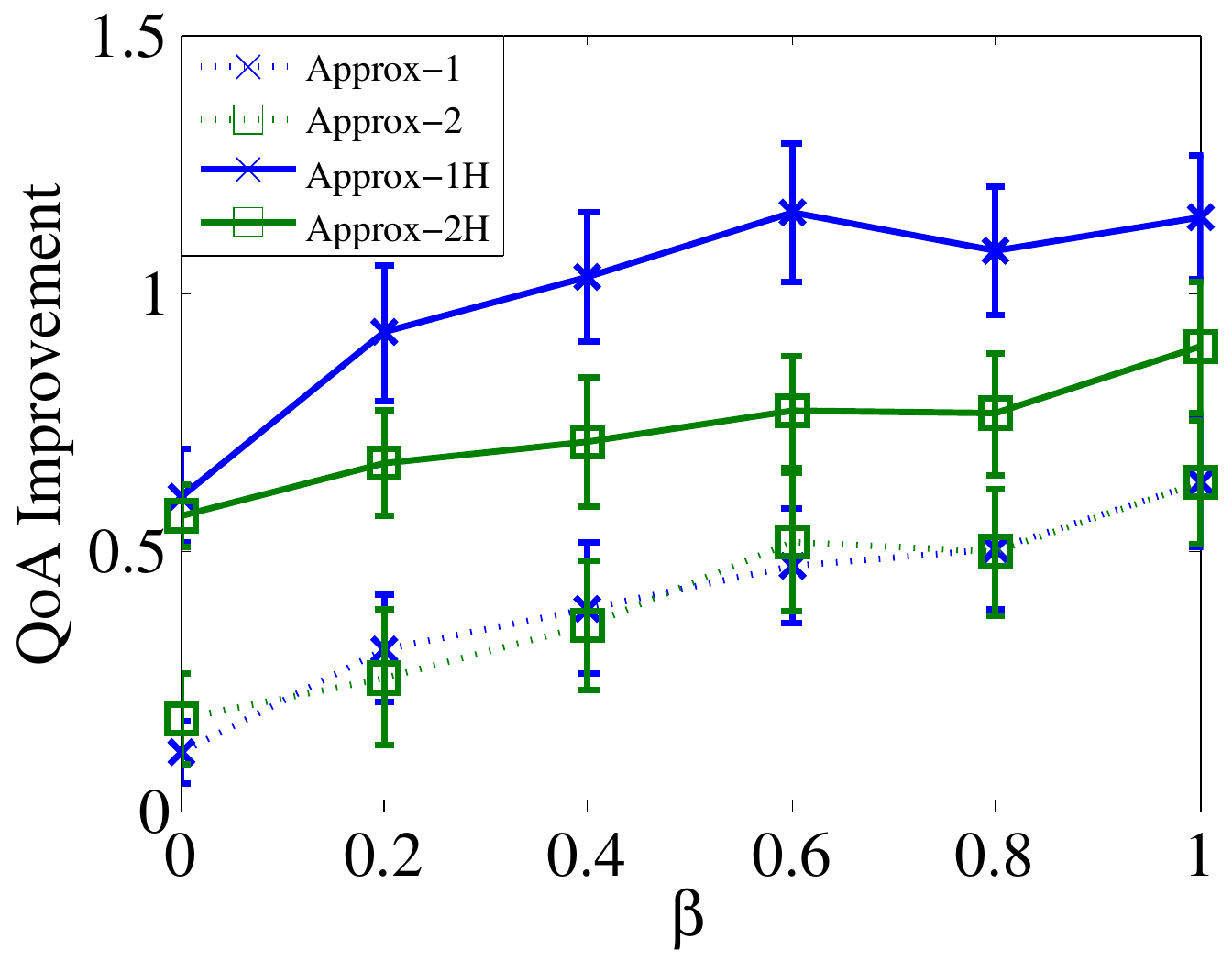}
		\caption{Improvement in quality of aggregation vs. $\beta$.}%
		\label{fig:beta}%
	\end{minipage}%
\end{figure*}

\subsection{Performance Comparison with the Optimal Solution}
In this section, we compare the performance of the proposed methods to the optimal 
solution. Since calculating optimal solution is computationally infeasible in large scale networks, we set up a small comparison experiment where 15 sensor nodes with communication range of 10m dispersed in a field with side length of 40m and sink coordinate is (20,40). Moreover, due to small network size, we consider all sensors as the source nodes.

Fig. \ref{fig:opt} portrays $QoA$ of markov based algorithms against sink deadline. The main purpose is to compare our schemes with the optimal. The result for the algorithms ``Approx-1'', ``Approx-2'', ``Approx-1H'' and ``Approx-2H'' are very close to each other. ``Approx-1H'' is 93\% close to optimal in this case which is slightly better than the other algorithms. We believe that in real-world scenarios with the higher number of sensor nodes, the performance difference between the markov based algorithms is more visible than that of the small scenario. To scrutinize this claim in more detail, we set up another set of experiments to investigate the improvements against various deadlines in the next subsection. 


\begin{figure*}[t!]
	\centering
	\vspace{-0.05em}
	\begin{subfigure}[b]{0.3\textwidth}
		\begin{center}
			\includegraphics[width=\textwidth]{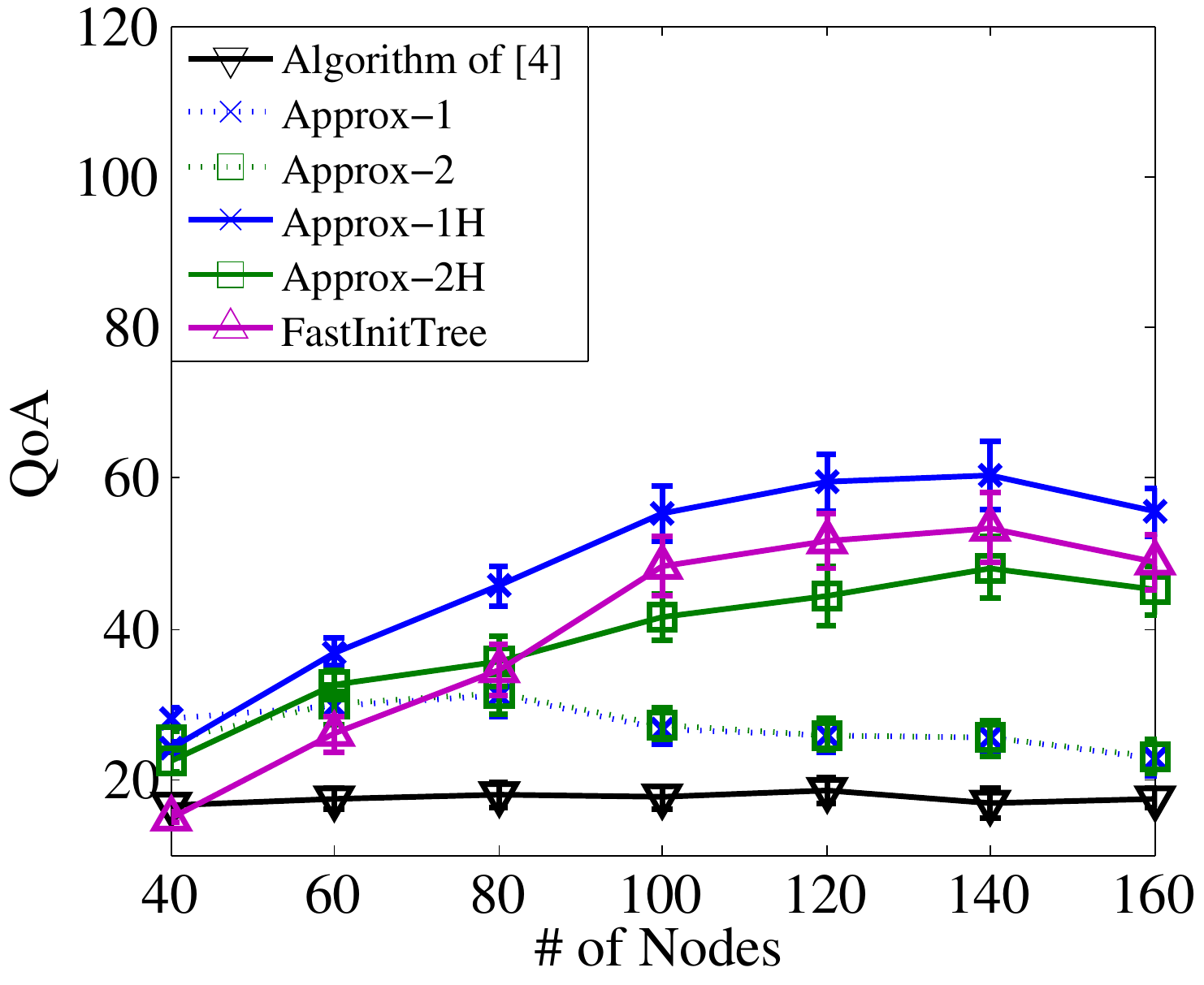}
			\caption{$D=10$}
			\label{fig:QoA-N-D10}
		\end{center}
	\end{subfigure}
	\begin{subfigure}[b]{0.31\textwidth}
		\begin{center}
			\includegraphics[width=\textwidth]{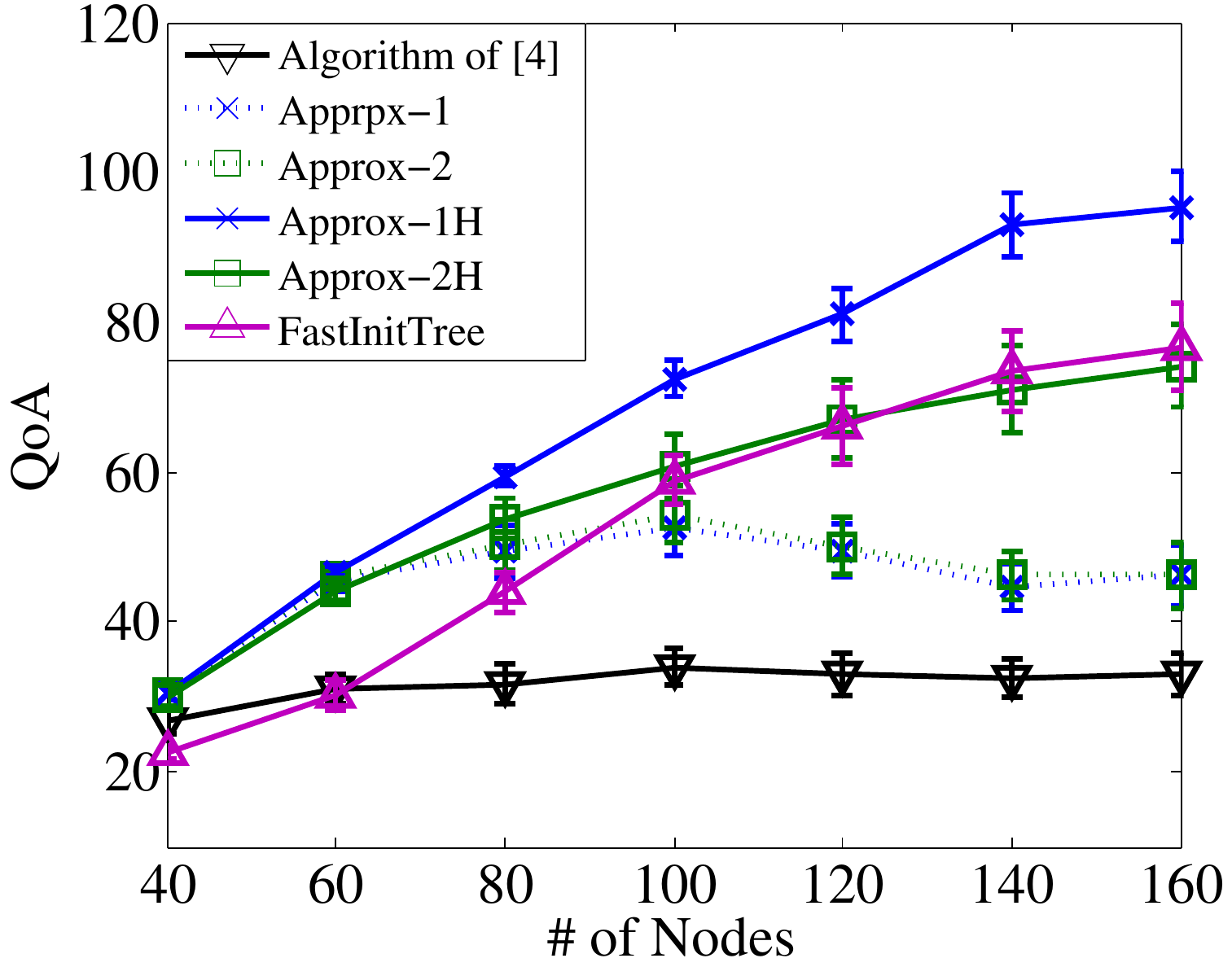}
			\caption{$D=15$}
			\label{fig:QoA-N-D15}
		\end{center}
	\end{subfigure}%
	\begin{subfigure}[b]{0.32\textwidth}
		\begin{center}
			\includegraphics[width=\textwidth]{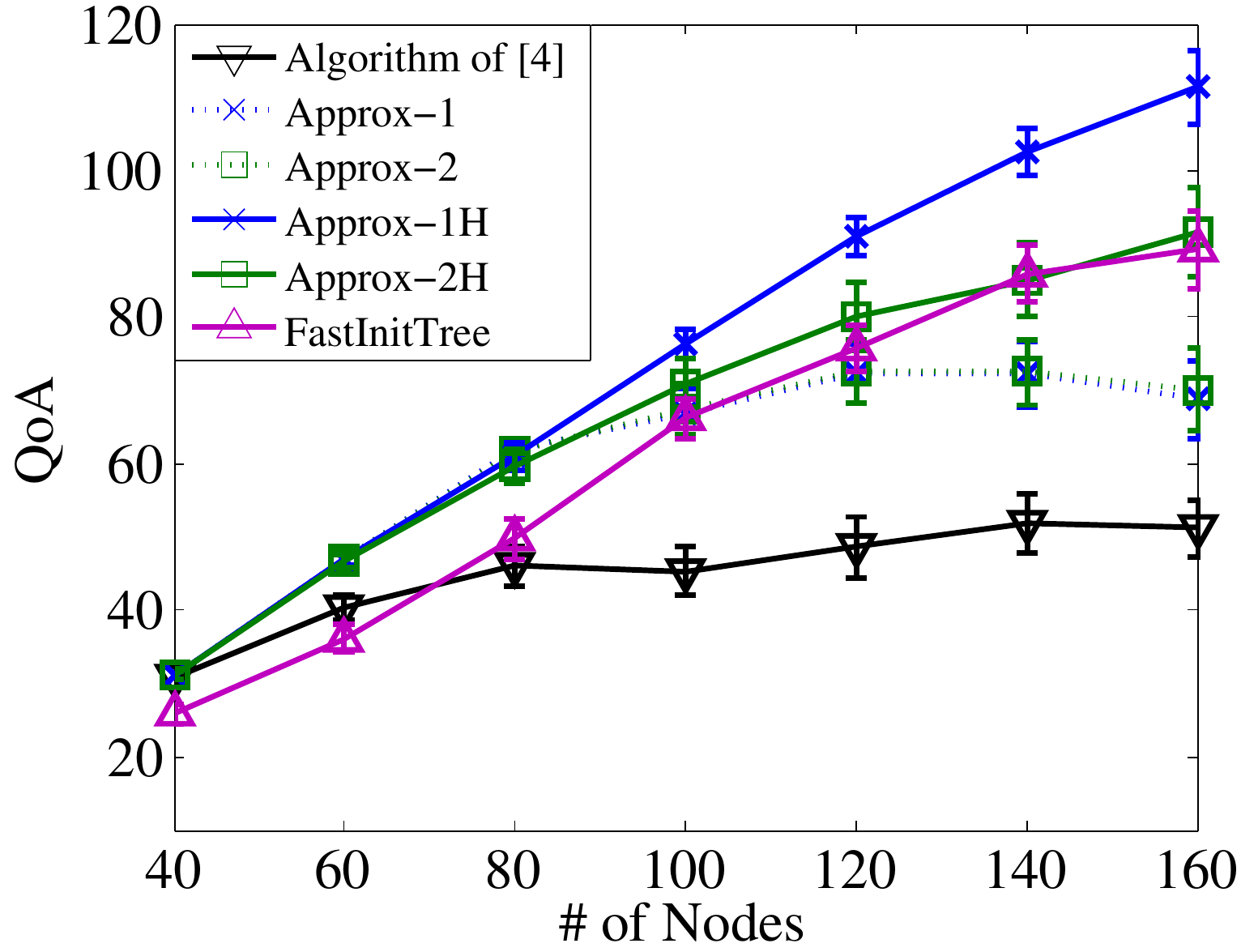}
			\caption{$D=20$}
			\label{fig:QoA-N-D20}
		\end{center}
	\end{subfigure}%
	\vspace{2mm}
	\caption{Quality of aggregation vs. network size.} 
	\label{fig:QoA-N-D}
\end{figure*}

\subsection{The Effect of the Deadline}
We now study the effect of sink deadline on $QoA$.
Based on Fig. \ref{fig:deadline}, the trend is that $QoA$ improves as deadline increases. This is in line with the fact that by increasing the deadline, more sensor nodes have the opportunity to participate in data aggregation. 

A notable observation is that \emph{FastInitTree} shows a better performance compared to both ``Approx-1'' and ``Approx-2'' when the deadline is less than 17. Its result is also 96\% of the ``Approx-1H''. There is also a small difference between ``Approx-1'' and ``Approx-2''. On average, ``Approx-2'' is 98\% close to ``Approx-1''. Based on these observations, \emph{FastInitTree} seems a proper choice with respect to its low overhead and low cost. ``Approx-1H'' has the best performance among all algorithms and improves ``Algorithm of \cite{Shroff}'' by 106\%, on average. The poor quality of ``Algorithm of \cite{Shroff}'' is a result of ignoring the impact of data aggregation tree structure. 

\subsection{The Effect of Parameter $\beta$}
As it is stated in Section IV.A, the approximation gap theoretically decreases while $\beta$ increases. This parameter is an input for the proposed Markov based approximation algorithms and has a big impact on convergence rate of the algorithms. We depict the effect of $\beta$ by simulation in Fig. \ref{fig:beta}. Since ``Algorithm of \cite{Shroff}'' and \emph{FastInitTree} are independent of the value of $\beta$, Fig. \ref{fig:beta} only portrays the amount of improvements by Markov based algorithms to ``Algorithm of \cite{Shroff}''.
By increasing $\beta$, in addition to achieving higher improvements, we observe that the improvement momentum of Markov based schemes to ``Algorithm of \cite{Shroff}'' degrades while $\beta$ grows. This is a consequence of fast convergence of approximation schemes to the optimal where in the proximity of optimal solution improvements are smaller. The experimental results of Fig. \ref{fig:beta} confirm the theory. 

\subsection{The Effect of the Network Size}
Fig. \ref{fig:QoA-N-D}a-\ref{fig:QoA-N-D}c depict obtained $QoA$ values for network sizes of 40 to 160 with step 20 for deadline values of 10, 15 and 20. An interesting observation is that the performance of \emph{FastInitTree} significantly increases as network size grows. As a result, ``Approx-1H'' and ``Approx-2H'' where use \emph{FastInitTree} as initial point, show the same trend. The reason behind depicting the results with different deadlines is to fairly compare \emph{FastInitTree} to ``Approx-1'' and ``Approx-2''. It can be observed from Fig. \ref{fig:QoA-N-D}a-\ref{fig:QoA-N-D}c that \emph{FastInitTree} achieves better $QoA$ when the ratio $N/D$ is greater than a specific value (here for $N/D>6$, approximately). Therefore, \emph{FastInitTree} works better with high values of deadline, in general.
 
The improvement by Markov based approximation algorithms to ``Algorithm of \cite{Shroff}'' are at least 39\%, 37\%, 32\% and 27\% in all network sizes for ``Approx-1'', ``Approx-2'', ``Approx-1H'' and ``Approx-2H'', respectively. Moreover, the average improvements are 52\%, 49\%, 68\% and 53\%, respectively. ``Approx-1H'' and ``Approx-2H'' show better performance because they start from initial tree constructed by \emph{FastInitTree} algorithm. The result for \emph{FastInitTree} algorithm is promising since it improves ``Algorithm of \cite{Shroff}'' by 29\%, on average. ``Algorithm of \cite{Shroff}'' shows little or no variation in obtained $QoA$ value while network size grows.

\subsection{The Effect of FastInitTree on the Convergence Rate}
Since the transition rates are set wisely to improve the maximum $QoA$ of the aggregation tree, we expect to obtain a better $QoA$ as the number of transitions increases. Each transition can only occur after a node's timer expiration. However, all timer expiration does not lead to a transition. A key point here is that a desired level of $QoA$ can be achieved with a fewer number of transitions if the initial tree provided by \emph{FastInitTree} algorithm is chosen wisely. Fig. \ref{fig:iteration} demonstrates how the four Markov-based approximation algorithms improve as the number of iterations increases. Note that the only difference between ``Approx-1'' and ``Approx-1H'' is in their initial state. The difference between ``Approx-2'' and ``Approx-2H'' is also similar. Then, a key point here is the effect of \emph{FastInitTree} algorithm on the convergence rate of Markov approximation. ``Approx-2H'' with 10 iterations reaches the same level of $QoA$ that ``Approx-2'' obtains with 40 iterations. In a similar case, ``Approx-1H'' with only 10 iterations works better than ``Approx-1'' after 40 iterations.

Finally, in a \textit{microscopic} view in Fig.~\ref{fig:transition}, we demonstrate the evolution of the maximum achieved $QoA$ after each transition, i.e., migrating to a new aggregation tree, for a randomly selected sample topology. 
%
%
\begin{figure}
\center
\begin{minipage}[b]{.32\textwidth}
\includegraphics[width=\textwidth]{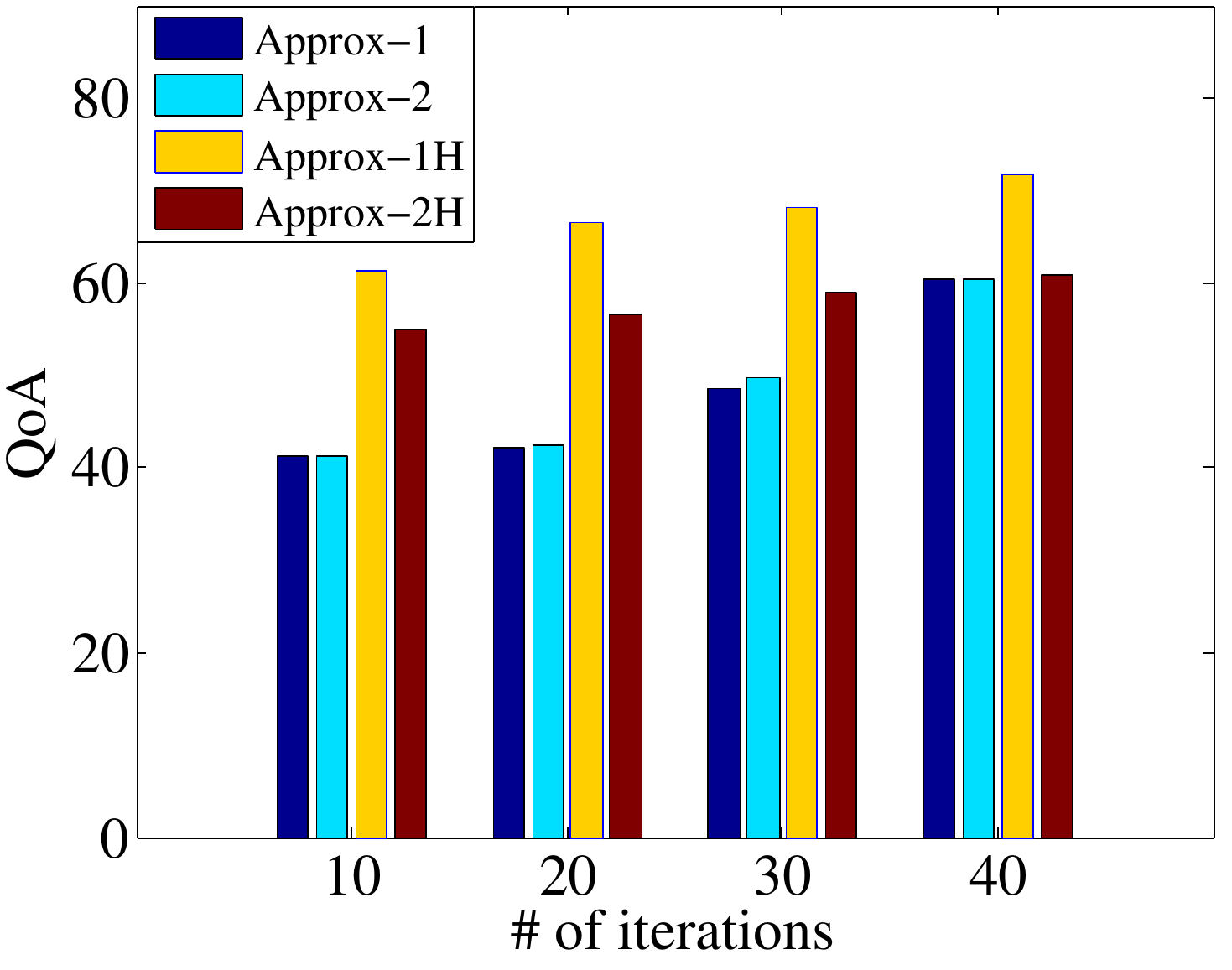}
\caption{$QoA$ vs. iteration numbers}%
\label{fig:iteration}%
\end{minipage}
\hspace{1mm}
\begin{minipage}[b]{.315\textwidth}
\includegraphics[width=\textwidth]{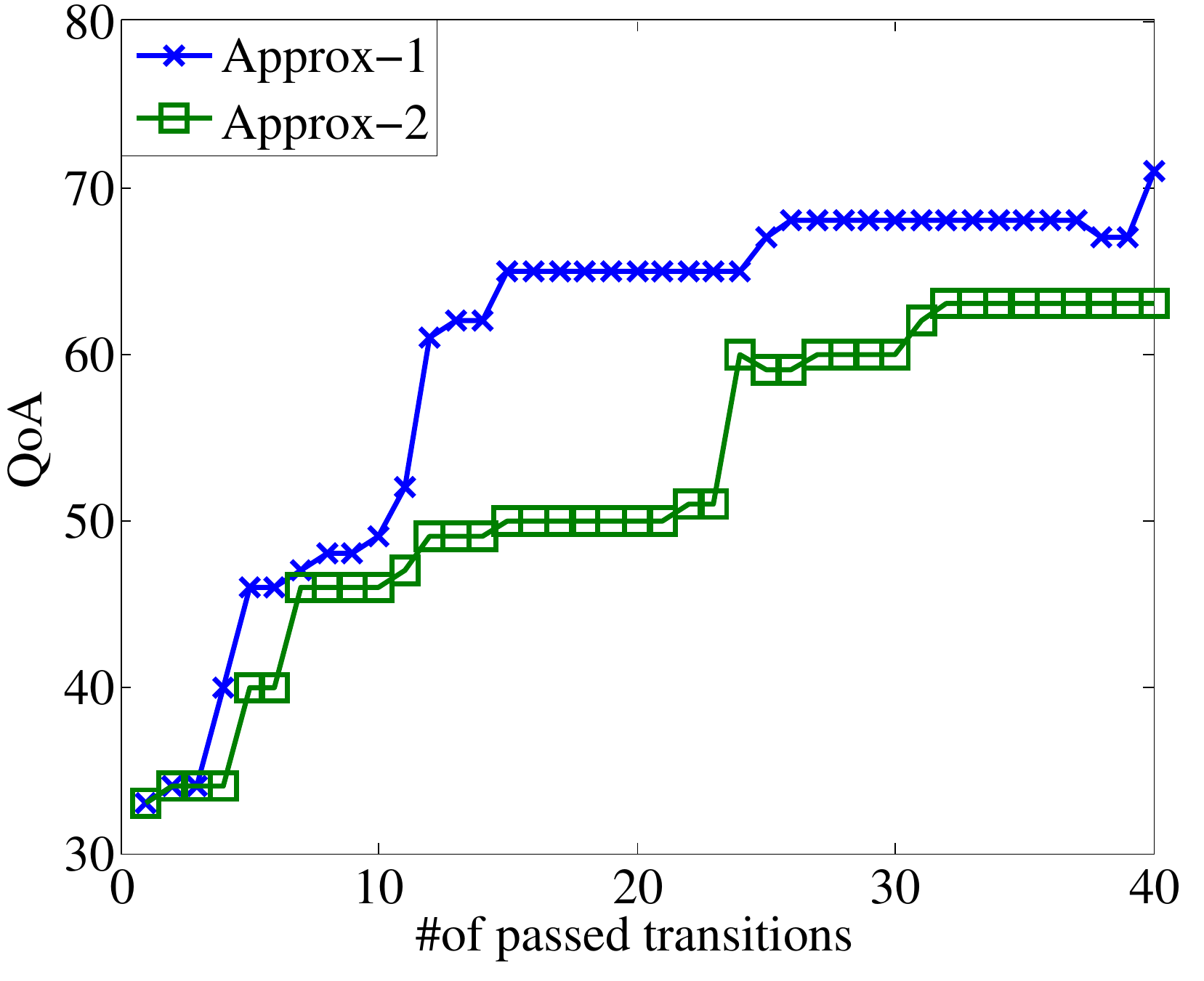}
\caption{Improvement of $QoA$ for a random topology}%
\label{fig:transition}%
\end{minipage}%
\end{figure}

\section{Conclusion}
\label{sec:conc}
In this paper, we addressed the NP-hard problem of constructing data aggregation tree in WSNs, with the goal of maximizing the number of nodes that the sink receives their data within an application-specific aggregation deadline. 
Two successive algorithms were proposed: \textit{first}, a distributed algorithm that runs in iterative manner and eventually converges to a bounded neighborhood of the optimum, and \textit{second}, a bootstrapping algorithm with low complexity that can be served as a good initial point for the former. Observations on experiments corroborated our analysis on the importance of constructing the optimal aggregation tree. Moreover, experimental results demonstrated that our methods not only achieved a close-to-optimal solution, but also significantly outperformed the existing methods that rely on fixed underlying data aggregation trees. 
Last but not the least, this work is the first attempt on leveraging Markov approximation as a general framework to tackle tree construction in WSNs and we believe that this solution approach can be used in several other applications for constructing trees in distributed manner. 

Obtained results open several important future directions. It would be interesting  to incorporate energy consumption and turn the problem to an energy-aware $QoA$ maximization one. This is important because the data aggregation is a periodic operation in the network and hence, relying on a fixed aggregation tree for a long time may lead to energy depletion of some specific nodes and degrade the network performance and lifetime. A wise policy might be to try to follow a uniform distribution of nodes' contribution in data aggregation, while keeping the $QoA$ at the desired level. The second line is to tackle forest construction problem for multi-sink networks. This is a challenging problem, since even the single sink scenario, as the special case of a multi-sink network, has been proved to be NP-hard.

\bibliographystyle{IEEEtranS}

\begin{thebibliography}{1}

\bibitem{Alinia15}
B. Alinia, M. H. Hajiesmaeili and, A. Khonsari, ``On the construction of maximum-quality aggregation trees in deadline-constrained WSNs'', in \emph{IEEE INFOCOM}, 2015.

\bibitem{Heinzelman}
W. Heinzelman, A. Chandrakasan, and H. Balakrishnan, ``An application specific protocol
architecture for wireless microsensor networks,'' \emph{IEEE Trans.  on Wireless Communications}, vol.1, no. 4, pp. 660-670, 2002.


\bibitem{Rajagopalan}
R. Rajagopalan and K.P. Varshney, ``Data aggregation techniques in sensor networks: a survey,''
\emph{IEEE Commun. Surveys Tutorials}, vol. 8, no. 4, pp. 48-63, 2006.

\bibitem{Shroff}
S. Hariharan, and N. B. Shroff, ``Maximizing aggregated information
in sensor networks under deadline constraints,'' \emph {IEEE
Trans. on Automatic Control}, vol. 56, no. 10, pp. 2369-2380,
2011.

\bibitem{Shroff2013}
S. Hariharan, Z. Zheng, and N. B. Shroff, ``Maximizing information in unreliable sensor networks under deadline and energy constraints,'' \emph{IEEE Trans. on Automatic Control,} pp. 1416-1429, 2013.

\bibitem{Li}
H. Li, C. Wu, Q.S. Hua, and F. Lau, ``Latency-minimizing data aggregation in wireless sensor networks under physical interference model,'' \emph{Ad Hoc Networks}, vol. 12,  pp. 52-68, 2014.

\bibitem{Xu}
X. Xu, X.-Y. Li, X. Mao, S. Tang, and S. Wang, ``A delay-efficient algorithm
for data aggregation in multihop wireless sensor networks,'' \emph{IEEE
Trans. on Parallel and Distributed Systems}, vol. 22, pp. 163-175, 2011.

\bibitem{Alinia}
B. Alinia, H. Yousefi, M. S. Talebi, and A. Khonsari, ``Maximizing quality of aggregation in delay-constrained wireless sensor networks'', {\it {IEEE Communications Letters} }, vol. 17, no. 11, pp. 2084-2087, 2013.

\bibitem{Nath}
S. Nath, P. Gibbons, B. Phillip, S. Seshan, and R.Z. Anderson, ``Synopsis diffusion for robust aggregation in sensor networks,'' in \emph{Proc. ACM SenSys}, 2004.

\bibitem{Chen2}
M. Chen, S. C. Liew, Z. Shao, and C. Kai, ``Markov approximation for
combinatorial network optimization,'' \emph {IEEE Trans. on Information Theory}, vol. 59, no. 10, pp. 6301-6327, 2013.

\bibitem{Chen}
X. Chen, X. Hu, and J. Zhu, ``Minimum data aggregation time problem in
wireless sensor networks,'' in \emph{Proc. IEEE MSN}, 2005.

\bibitem{Wan}
P.J. Wan, S.C.H. Huang, L. Wang, Z. Wan, and X. Jia, ``Minimum-latency
aggregation scheduling in multihop wireless networks,'' in \emph{Proc. ACM MOBIHOC}, 2009.

\bibitem{Li2}
X.Y. Li, X. Xu, S. Wang, S. Tang, G. Dai, J. Zhao, and Y. Qi, ``Efficient data
aggregation in multi-hop wireless sensor networks under physical
interference model,'' in \emph{Proc. IEEE MASS}, 2009.

\bibitem{Guo}
L. Guo, Y. Li, and Z. Cai ``Minimum-latency aggregation scheduling in wireless sensor network,''  \textit{Journal of Combinatorial Optimization}, 2014.

\bibitem{Shroff2011}
S. Hariharan, and N. B. Shroff, ``Deadline constrained scheduling for data aggregation in unreliable sensor networks,'' in \emph{Proc. IEEE WiOpt}, 2011.

\bibitem{Shroff2014}
Z. Zheng, and N. B. Shroff, ``Submodular utility maximization for deadline constrained data collection in sensor networks'', {\it {IEEE Trans. on Automatic Control} }, vol. 59, no. 9, pp. 2400-2412, 2014.


\bibitem{Tan}
H.X. Tan, M.C. Chan, W. Xiao, P.Y. Kong, and C.K. Tham, ``Information quality aware routing in event-driven sensor networks,'' in \emph{Proc. IEEE INFOCOM}, 2010.

\bibitem{Wu}
Y. Wu, Z. Mao, S. Fahmy, and N. B. Shroff, ``Constructing maximum lifetime
data-gathering forests in sensor networks'' \emph{IEEE/ACM Trans.
on Networking}, vol. 18, no. 5, pp. 1571-1584, 2010.

\bibitem{Li3}
D. Li, J. Cao, M. Liu, and Y. Zheng, ``Construction of optimal data aggregation trees for wireless sensor networks,'' in \emph{Proc. IEEE ICCCN}, 2006.

\bibitem{Fahmy}
Y. Wu, S. Fahmy, and N.B. Shroff, ``On the construction of a maximum-lifetime data gathering tree in sensor networks: NP-completeness and approximation algorithm,'' in \emph{Proc. IEEE INFOCOM}, 2008.

\bibitem{Kuo}
T. W. Kuo, and M.J. Tsai, ``On the construction of data aggregation tree with minimum energy cost in wireless sensor networks: NP-completeness and approximation algorithms,'' in \emph{Proc. IEEE INFOCOM}, 2012.

\bibitem{Shan}
M. Shan, G. Chen, D. Luo, X. Zhu, and X. Wu, ``Building Maximum Lifetime Shortest Path Data Aggregation Trees in Wireless Sensor Networks,'' \emph{ACM Transactions on Sensor Networks (TOSN)}, 2014.


\bibitem{Chekuri}
C. Chekuri, and A. Kumar, ``Maximum coverage problem with group budget constraints and applications,''  \emph{Approximation, Randomization, and Combinatorial Optimization. Algorithms and Techniques}, vol. 3122, pp. 72-83, 2004.

\bibitem{Zhang}
S. Zhang, Z. Shao, M. Chen, and L. Jiang, ``Optimal distributed P2P streaming under node degree bound,'' \emph{IEEE/ACM Trans. on Networking}, vol. 22, no. 3, June 2014.

\bibitem{Diaconis}
P. Diaconis and D. Stroock, ``Geometric bounds for eigenvalues of
Markov chains,'' \emph{The Annals of Applied Probability}, pp. 36–61, 1991.

\bibitem{Shao2013}
Shao, Ziyu, Xin Jin, Wenjie Jiang, Minghua Chen, and Mung Chiang. ``Intra-data-center traffic engineering with ensemble routing,'' in \textit{Proc. INFOCOM}, 2013.

\bibitem{Krishnamachari2}
B. Krishnamachari, D. Estrin, and S. Wicker, ``Modelling data-centric routing in wireless sensor networks,'' in \emph{Proc. IEEE INFOCOM}, 2002.




\end{thebibliography}

\end{document}